\documentclass[12pt]{article}

\usepackage{fullpage}
\usepackage[T1]{fontenc}
\usepackage{ae,aecompl}
\usepackage[dvips]{graphicx}
\usepackage{amssymb}
\usepackage{amsmath}
\usepackage{subfigure}
\usepackage[round,authoryear]{natbib}
\usepackage{color}
\usepackage{array}
\usepackage{rotating}
\usepackage{colortbl}
\usepackage{tikz}
\usetikzlibrary{patterns}

\usepackage[normalem]{ulem}
\usepackage{cancel}
\definecolor{webgreen}{rgb}{0,0.4,0}
\definecolor{webbrown}{rgb}{0.6,0,0}
\definecolor{purple}{rgb}{0.5,0,0.25}
\definecolor{darkblue}{rgb}{0,0,0.7}
\definecolor{darkred}{rgb}{0.7,0,0}
\definecolor{darkgreen}{rgb}{0,0.7,0}
\usepackage[pdfborder=false]{hyperref}
\hypersetup{colorlinks,citecolor=darkred,filecolor=black,linkcolor=darkblue,urlcolor=webgreen,pdfpagemode=none,
pdfstartview=FitH}
\newcommand{\ignore}[1]{}
\newtheorem{lemma}{{\sc Lemma}}

\newtheorem{theorem}{{\sc Theorem}}
\newtheorem{defn}{{\sc Definition}}

\newtheorem{claim}{{\sc Claim}}

\newenvironment{proof}{\noindent {\bf \sl Proof\/}:\enspace}
{\hfill $\blacksquare{}$ \vspace{12pt}}
\usepackage{sectsty}
\sectionfont{\centering\normalfont\scshape}
\subsectionfont{\centering\normalfont}

\begin{document}

\allowdisplaybreaks

\begin{titlepage}
\title{\Large{\textbf{Selling to a manager and a budget-constrained agent}}\thanks{
A slightly different version of this paper is part of the Ph.D. thesis of 
the second author submitted to the Indian Statistical Institute. 
This paper supersedes an earlier working paper
titled ``selling to a naive (agent, manager) pair'' with significant differences in model and results. We are grateful to two anonymous referees and an associate editor for detailed comments that led to significant improvements in the paper. We are also graetful to 
Sushil Bikhchandani, Francis Bloch, Abhinash Borah, Juan Carlos Carbajal, Rahul Deb, Bhaskar Dutta,
Yoram Harlevy, Sridhar Moorthy, Arunava Sen, Ran Speigler, Rakesh Vohra, and seminar participants
at several conferences for helpful comments on the earlier version
which led us to think of the model in the current paper. Debasis Mishra acknowledges financial support from the Science and Engineering Research Board (SERB Grant No. SERB/CRG/2021/003099) of India.}}
\author{Debasis Mishra\thanks{Indian Statistical Institute, Delhi ({\tt dmishra@isid.ac.in}).}$\;\;$ and Kolagani Paramahamsa\thanks{Indian Statistical Institute, Delhi ({\tt kolagani.paramahamsa@gmail.com}).}}

\maketitle

\begin{abstract}

We analyze a model of selling a single object to a manager-agent pair who want to acquire the object for a firm.
The manager and the agent have different assessments of the object's value to the firm.
The agent is budget-constrained while the manager is not. The agent participates in the mechanism, but she can (strategically) approach the manager for decision-making. We derive the revenue-maximizing
mechanism in a two-dimensional type space (values of the agent and the manager). We show that below a threshold budget,
a mechanism involving two posted prices and three outcomes (one of which involves randomization) is
the optimal mechanism for the seller. Otherwise, a single posted price mechanism is optimal.



\end{abstract}
\thispagestyle{empty}
\end{titlepage}

\section{Introduction}
\label{sec:intro}

A seller is selling an object to a firm.
An agent (a representative of the firm) participates in a mechanism
to acquire the object. The agent's payment
cannot exceed a budget. However, the agent can approach the manager 
(the board of directors), who is not budget-constrained, for decision-making.
When approached, the manager takes a decision based on his own preference 
over outcomes. Otherwise, the agent takes a decision based on her 
preference (which can be different from the manager) over outcomes but 
respecting the budget constraint.

The object is of value to the firm, where the manager and the
agent are shareholders. So, they both want to maximize the firm's payoff
from acquiring the object. However, they evaluate the value of the object to the firm differently.
In particular, we assume that the agent finds more value in the object than the manager. The difference in valuation may be because the agent gets some additional personal value when the firm acquires the object  or because the agent is more informed about the potential uses of the object.
We assume that the value to the manager is common knowledge
among the manager and the agent, but the value to the agent is privately known to her.\footnote{The common
knowledge of the value of the manager among the manager-agent pair may be
because the agent knows some common attributes of the object that
the manager uses to evaluate the object. On the other hand, the manager does not know other uses of the object or the personal value of the object to the agent.
The agent cannot persuade the manager about the extra value of the object.}
The seller does not know the values of the manager and the agent but observes the agent's
budget constraint. The agent's budget constraint
may be reflected by the liquidity constraint of the firm, which may be verified from publicly
available information. \footnote{Some of the results can be extended 
to the case with private budget, but requires significantly longer analysis. 
\citet{P21} contains the details.}

Our agent approaches the manager for decision-making if the best outcome from the menu of
the mechanism in her {\sl budget set} is worse than the manager's best outcome from
the menu of the mechanism. Otherwise, the agent chooses the best outcome from the menu of
the mechanism in her budget set.

What is the expected revenue-maximizing mechanism of the seller? The communication 
between the manager and the agent (when the budget constraint binds for the 
agent) has different implications on incentive constraints than in the usual 
mechanism design problems. We show that a revelation principle holds and a 
new class of (direct) mechanisms, which we call {\sc post-2} mechanisms, are 
incentive compatible.

In a {\sc post-2} mechanism, the seller posts two prices: $\kappa_1,\kappa_2$, 
both above the budget $b$. If the agent has value less than $\kappa_1$, the 
object is not allocated (and zero payment
is made). Else, a fraction $\frac{b}{\kappa_1}$ of the object is allocated at 
a price equal to the budget $b$.\footnote{If the object is indivisible, this 
fraction is the probability of getting the object. If the object is divisible, this 
fraction is interpreted as the share of the object allocated.} Further, if the
manager has value more than $\kappa_2$ (and the agent has value more than 
$\kappa_1$), the remaining
fraction of the object $\big(1-\frac{b}{\kappa_1}\big)$ is allocated at a price 
$\big(1-\frac{b}{\kappa_1}\big)\kappa_2$.
This mechanism is akin to a two-part tariff mechanism.

A simpler class of mechanisms ignores the value of the manager and
considers only the value of the agent. By definition, such a mechanism cannot charge more than the budget. A {\sc post-1} mechanism is such a posted price mechanism with a price $\kappa \le b$.

Our main result says that there is a threshold $b^*$
such that for all budgets $b \le b^*$, the optimal mechanism is a {\sc post-2} mechanism and for all $b > b^*$,
the optimal mechanism is a {\sc post-1} mechanism. The threshold corresponds to the optimal posted-price
of an {\sl unconstrained agent} whose values are drawn using the marginal distribution of the current model, i.e.,
$b^*$ is the optimal solution to $\max x(1-F_1(x))$, where $F_1$ is the marginal distribution of the value
of the agent. For all our results, we assume that $x(1-F_1(x))$ is strictly concave, an assumption
satisfied by many distributions (which also allow for correlation among values of the manager and the agent).
This allows us to identify $b^*$ uniquely. Under stronger conditions on distributions, we give a more
precise description of the optimal mechanism.

These results are in contrast with the single object mechanism design literature, where
a deterministic posted-price mechanism is optimal~\citep{Mussa78}. The deterministic optimality result
usually does not extend to multidimensional mechanism design problems. Even for two-dimensional mechanism
problems, the menu in the optimal mechanism may contain an infinite set of outcomes~\citep{HN19}. On the other hand, ours is a
two-dimensional mechanism design problem, and the optimal mechanism has three outcomes in the menu. The particular
nature of decision-making allows tractability in our two-dimensional model and results in a simple solution. To elaborate on this a bit further, the decision-making of the manager and the agent are separable in some sense: the agent can take a decision if her optimal outcome in her budget set is better than the optimal outcome of the manager. This crucial assumption allows us to first optimize over all mechanisms in which all types pay less than the budget. Standard techniques lead to a posted price mechanism, which is our {\sc post-1} mechanism, to be optimal. If we look at a mechanism where some types pay more than the budget, we can partition the type space into three parts (each part corresponding to either the agent or the manager taking the decisions). We then establish upper bounds on revenue from a mechanism in each part, and show that this upper bound can be achieved by a {\sc post-2} mechanism. Summarizing, we overcome the difficulties of multidimensional screening problems by generating incentive compatibility constraints through a sequence of optimization over each dimension. This may be useful in future in considering general multidimensional screening problems.


\subsection{An illustration}
\label{sec:illust}

Suppose the values of the agent ($v_1$) and the manager ($v_2$) are
distributed in the triangle: $\{(v_1,v_2) \in [0,1]^2: v_1 \ge v_2\}$.
Consider a simple posted-price mechanism $p > b$ such that the object is allocated at price $p$ if
the manager has value more than $p$. Otherwise, the object is not allocated and zero payment is
made. The left triangle in Figure \ref{fig:illu} illustrates the mechanism\footnote{Here in the outcomes $(1,p)$, $(\frac{b}{p},b)$ etc the first argument is the allocation probability and the second is payment made to the seller.}.
This mechanism is incentive compatible (IC) since the agent and the principle prefer getting the object when $v_1 \geq v_2 \geq p$. The agent approaches the manager for decision-making in this case (since $p > b$) and chooses not to buy otherwise.

However, the seller can improve expected revenue from this mechanism by the following mechanism.
The seller does not allocate the object if the value of the agent is less than $p$.
Types $(v_1,v_2)$ with $v_1 \ge p$ but $v_2 < p$ get the object with probability $\frac{b}{p}$
and pay $b$. Types $(v_1,v_2)$ with $v_1 \ge v_2 \ge p$ get the object with probability $1$ and pay $p$.
So, the seller allocates the object more often in this mechanism and collects more revenue from
types $(v_1,v_2)$ with $v_1 \ge p$ and $v_2 < p$. The right triangle in Figure \ref{fig:illu} illustrates this mechanism.
\begin{figure}[!hbt]
\centering
\includegraphics[width=4.5in]{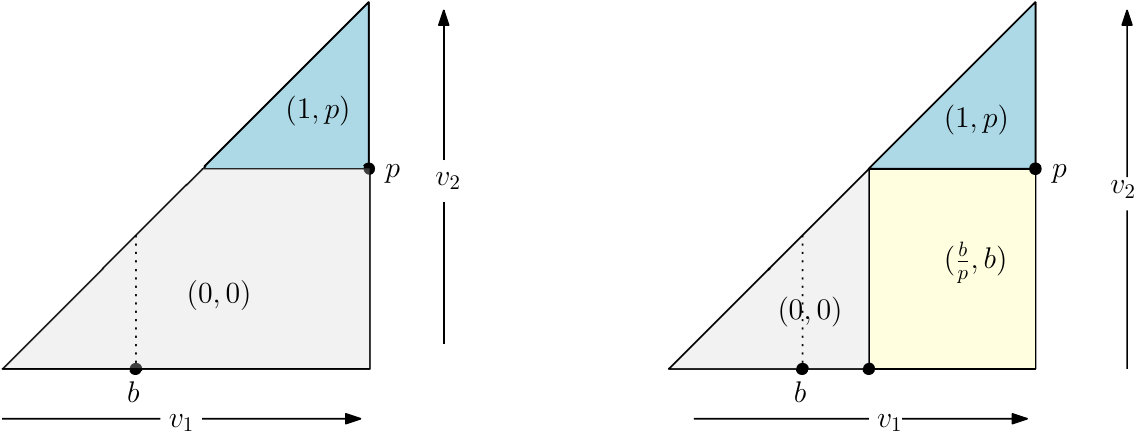}
\caption{Illustration of new IC mechanisms}
\label{fig:illu}
\end{figure}

Why is this mechanism incentive compatible? There are three types of manager-agent pairs corresponding
to the three outcomes in the menu of the mechanism. When $v_2 \le v_1 < p$, both the agent and the manager prefer the
outcome of not getting the object and zero payment to other outcomes in the menu. Similarly, when
$p \le v_2 \le v_1$, both the agent and the manager prefer getting the object at a price $p$ to
the other outcomes in the menu. In this case, the agent approaches the manager for decision-making (since $p > b$).

The types $(v_1,v_2)$ with $v_1 \ge p > v_2$ require a careful look.
At any such type $(v_1,v_2)$, the agent's best outcome in her {\sl budget set} is
to get the object with probability $\frac{b}{p}$ at price $b$, since $\frac{b}{p}(v_1-p) \ge 0$.
However, the agent prefers the outcome outside her budget set to this outcome: $v_1 - p \ge \frac{b}{p}(v_1-p)$
with strict inequality holding if $v_1 > p$. But she knows that the best outcome for the manager is to
not get the object and pay zero: $v_2 - p < \frac{b}{p} (v_2 - p) < 0$. So, the agent is strategic
and chooses the best outcome in her budget set. This makes the mechanism incentive compatible.

It is worth noting that when $v_2 < p \le v_1$, the manager gets a negative payoff in the mechanism.
Our interpretation is that the manager incurs a cost in making decisions. 
On average, considering this cost, the manager’s ex-ante payoff is positive. 
Hence, the manager is happy to delegate decision-making to the agent with a budget constraint.
We do not model this explicitly but discuss models in the literature that study such contracts between the manager and the agent. 
In Section \ref{sec:disc}, we illustrate with an example that the manager’s ex-ante payoff is positive in our optimal mechanism.

The seller is able to exploit this by offering an additional outcome in the menu of the mechanism
that the agent will find optimal to accept at some types. This allows the seller to extract more revenue.
Though there are other mechanisms that can allow the seller to extract more revenue than a posted-price mechanism, we show that
an optimal mechanism is a simple one (like in Figure \ref{fig:illu} right triangle). Thus, our result says that
there is an upper bound (achieved by our optimal mechanism) on the revenue extracted
by the seller in these decision making environments.

\subsection{Layout of the paper}

The rest of the paper is organized as follows. Section \ref{sec:model} describes the model.
Section \ref{sec:opt} has our main results characterizing the optimal mechanism. Section \ref{sec:lit}
relates our model and results to the literature. Finally, we end with some discussions of our modeling
assumptions in Section \ref{sec:disc}. All our missing proofs and the discussions on our revelation principle
are in an appendix at the end.

\section{The model}
\label{sec:model}

We now formally introduce our model. There is a single object which is sold by a seller to a firm.
An agent and her manager are the decision-makers for the firm. We describe the preferences of the agent and the manager
and the sequence of decison-making in our model. We index the agent by $1$ and the manager by $2$ in
the paper.

{\bf Preferences.} Let $Z$ be the set of all outcomes, i.e., $Z:=\{(a,t): a \in [0,1], t \in \mathbb{R}\}$, where
$a$ is allocation probability and $t$ is the payment made to the seller.
Let $b$ be the budget of the agent. That is, the agent cannot pay more than $b$. Throughout the paper, we will assume that $b$ is
common knowledge among the agent, the manager and the seller.
For every $X \subseteq Z$, define $X_b:=\{(a,t) \in X: t \le b\}$ to be the {\sl budget set} of
the agent.

We now introduce some notation to define how the manager-agent pair chooses
outcomes from a given set of outcomes. Let $v_i \in [0,\beta]$ be any valuation, where $\beta \in \mathbb{R}_{++}$.
Define for every valuation $v_i$ and every $X \subseteq Z$,
\begin{align*}
Ch(X;v_i) &= \{(a,t) \in X: v_i a - t \ge v_i a' - t'~\forall~(a',t') \in X\}
\end{align*}
Note that $Ch$ is a standard choice correspondence satisfying independence of irrelevant alternatives. That is, if
$X' \subseteq X$ and $Ch(X;v_i) \subseteq X'$, then $Ch(X';v_i)=Ch(X;v_i)$.

For any pair of sets of outcomes $X,Y \subseteq Z$, we say
\begin{align*}
X \unrhd_{v_i} Y~\textrm{if}~v_i a - t \ge v_ia' - t'\qquad~\forall~(a,t) \in X,~\forall~(a',t') \in Y
\end{align*}
We write $X \rhd_{v_i} Y$ if the above inequality is strict at least for one pair of outcomes.
If $X \unrhd_{v_i} Y$ and $Y \unrhd_{v_i} X$,  then we say $X \sim_{v_i} Y$. Note that for any $v_i \in [0,\beta]$,
the relations $\rhd_{v_i}$ and $\unrhd_{v_i}$ are transitive but incomplete. These relations will be
used to define choices of our manager-agent pair.

The type is a pair of valuations $v \equiv (v_1,v_2)$, where
$v_1$ is the value of the agent and $v_2$ is the value of the manager.
We assume that the value of the agent is higher than that of the manager. So, the type space is
\begin{align*}
V &= \{(v_1,v_2) \in [0,\beta]^2: v_1 \ge v_2\}
\end{align*}

For a vector of valuations $v = (v_1,v_2)$, the choice from
any $X \subseteq Z$ is given by
\begin{align*}
Ch(X;v) =
\begin{cases}
Ch(X;v_2) &~\textrm{if}~Ch(X;v_2) \rhd_{v_1} Ch(X_b;v_1) \\
Ch(X_b;v_1) &~\textrm{otherwise}
\end{cases}
\end{align*}
Note that if $Ch(X;v_2) \rhd_{v_1} Ch(X_b;v_1)$ is not true, then either $Ch(X;v_2) \sim_{v_1} Ch(X_b;v_1)$ or there exists an outcome $(a',t') \in Ch(X;v_2)$
such that $v_1 a - t > v_1 a' - t'$ for all $(a,t) \in Ch(X_b;v_1)$.\footnote{In the above definition $Ch(X;v)$ refers to a two-dimensional $v \equiv (v_1,v_2)$
but the second argument of $Ch(X;v_2)$ (and $Ch(X_b;v_1)$) is one-dimensional. This abuse of notation
saves us from introducing a new notation.}

The choice correspondence $Ch(X;v)$ describes how choice is made by the manager-agent pair
from a set of outcomes $X$. In particular, if $v$ is the type vector, then the agent first evaluates
her choice correspondence from the budget set: $Ch(X_b;v_1)$. She then evaluates the (unconstrained) choice
correspondece of the manager: $Ch(X;v_2)$. The evaluation of $Ch(X;v_2)$ critically uses the
fact that the agent knows the value $v_2$ of the manager. Now, the agent compares
$Ch(X;v_2)$ and $Ch(X_b;v_1)$ using her own value $v_1$. If the unconstrained correspondence of the
manager is strictly better than the constrained correspondence of the agent according to $\rhd_{v_1}$,
she approaches the manager for decision-making.\footnote{We could allow 
the agent to approach the manager when the unconstrained correspondence of 
the manager is {\sl weakly better} than her constrained correspondence. Though this
complicates the analysis, the conclusions of the paper do not change.} In that case, we set $Ch(X;v)$ equal to
$Ch(X;v_2)$. Else, the agent makes the decision and her constrained choice correspondence $Ch(X_b;v_1)$
equals $Ch(X;v)$. The choice correspondence $Ch(X;v)$ denotes how choices are made in our
model. It is not difficult to see that $Ch(X;v)$ need not satisfy independence of irrelevant alternatives.

{\bf Mechanism.} A mechanism is a pair $(q,p)$, where $q:V \rightarrow [0,1]$ is the allocation rule
and $p:V \rightarrow \mathbb{R}$ is the payment rule. The range of the mechanism is $R(q,p):=\{(q(v),p(v)): v \in V\}$.
A mechanism $(q,p)$ is {\bf incentive compatible (IC)} if
\begin{align*}
\big(q(v),p(v)\big) &\in Ch(R(q,p);v)~\qquad~\forall~v \in V
\end{align*}
This notion of IC is standard in the literature of behavioral mechanism design~\citep{De14}.
Since the choice correspondence $Ch$ may fail independence of irrelevant alternatives,
it is not clear if restricting attention to direct mechanisms is without loss of generality.
For instance, \citet{Sa11} shows that the revelation principle may fail in models with {\sl behavioral} agents.
In Appendix \ref{sec:rev}, we show that a version of the revelation principle holds in our setting, and hence, it
is without loss of generality to focus attention to such direct mechanisms.

A mechanism $(q,p)$ is {\bf individually rational (IR)} if $(0,0) \in R(q,p)$.

It is important to emphasise the timing of the game induced by the direct mechanism.
\begin{itemize}

\item The seller announces the mechanism $(q,p)$. Let $X \equiv R(q,p)$ be its range.

\item The agent learns her value $v_1$ and the manager's value $v_2$.

\item If $Ch(X;v_2) \rhd_{v_1} Ch(X_b;v_1)$, then

\begin{itemize}

\item the agent approaches the manager for decision-making

\item the manager learns his type $v_2$ and chooses an outcome in $Ch(X;v_2)$

\end{itemize}

\item If $Ch(X;v_2) \rhd_{v_1} Ch(X_b;v_1)$ is {\sl not true}, then the agent chooses an outcome in $Ch(X_b;v_1)$.

\end{itemize}

As highlighted earlier in the introduction, there are couple of features of
this model that make it tractable. First, the agent learns her own value and that of the
manager. This can be interpreted in two ways. In the first interpretation, the agent knows all possible uses
of the object to the firm but also knows the possible uses of the object that the manager (board of directors)
knows. An alternate interpretation is that the value of the manager $v_2$ is the
actual value of the object to the firm, which is common knowledge between the
manager and the agent. But the agent gets additional payoff (for instance, reputation payoffs)
from acquiring the object which is privately known to her.
For instance, suppose the object has a value
$v_2$ to the firm. Then, the manager gets a payoff of $\alpha_2(v_2-p)$ by acquiring
the object at price $p$, where $\alpha_2$ is the share of the manager in the firm.
Thus, $v_2$ will be the value of the object to the manager. The agent gets a payoff
of $\alpha_1 (v_2-p) + \delta_1$, where $\alpha_1$ is the share of the agent in the firm
and $\delta_1$ is the additional payoff of the agent. Hence, we can
interpret $v_1:=v_2 + \frac{\delta_1}{\alpha_1}$ as the value to the agent.

Second, there is no formal contract between the manager and the agent about how the
agent will behave in the mechanism. Models studied in \citet{Burkett15,Burkett16} look at
such contracts between the manager and the agent given a mechanism. In our model, the
manager is happy to let the agent make decisions in the mechanism as long as the liquidity
constraint of the firm is not violated.

Finally, the only communication between the manager and the agent in our model is when the
agent approaches the manager for decision-making. In the interpretation where $v_1$ and $v_2$
are assessments of values by the agent and the manager respectively, and the agent knows
both of them, the agent does not {\sl persuade} the manager about the value of the object to
the firm. That is, there is no deliberation between the manager and the agent about the actual
value of the object to the firm. Models of mechanism design with such a communication phase is
studied in \citet{MT19}. In the second interpretation, where $v_2$ is treated as the
actual value of the object to the firm and $v_1-v_2$ is the additional private value of the agent,
our model assumes that the manager cannot question the decision making of the agent ex-post.

{\bf Prior.} The joint density function of $v \equiv (v_1,v_2)$ is
$f$ with support $V$. As the support of $v_2$ depends on the realized value of
$v_1$ (since each $v$ satisfies $v_1 \ge v_2$), the values $v_1$ and $v_2$ are not independent. We will denote the cdf of
$v$ as $F$. For any IC and IR mechanism $(q,p)$, the {\bf expected revenue}
is given by
\begin{align*}
\textsc{Rev}(q,p) &= \int \limits_V p(v)f(v)dv
\end{align*}

A mechanism $(q^*,p^*)$ is {\bf optimal} if it is IC and IR and for every other
IC and IR mechanism $(q,p)$, we have $\textsc{Rev}(q^*,p^*) \ge \textsc{Rev}(q,p)$: it is {\bf strictly optimal} if this inequality is strict for every other IC and IR mechanism.

We will denote the {\sl marginal distribution} of
agent's valuation as $F_1$ (which admits a density $f_1$) and manager's valuation as $F_2$ (which admits a density $f_2$).
The density functions $f_1$ and $f_2$ will be assumed to be positive and differentiable in $V$.

\section{Optimal mechanism}
\label{sec:opt}

We describe our main results in this section.
We start by describing two simple classes of mechanisms. Our main result will show
that the optimal mechanism belongs to one of these classes of mechanisms.
The first class contains posted-price
mechanisms for the agent.
\begin{defn}
A mechanism $(q,p)$ is a \textsc{post-1} mechanism if there exists $\kappa_1 \in [0,b]$ such that
for every $v \in V$, we have
\begin{align*}
(q(v),p(v)) =
\begin{cases}
(0,0) & \textrm{if}~v_1 \le \kappa_1 \\
(1,\kappa_1) & \textrm{otherwise}
\end{cases}
\end{align*}
\end{defn}
A {\sc post-1} mechanism is IC and IR. It only considers valuation of the agent.
The following class of mechanisms considers the valuations of the agent and the manager.
\begin{defn}
A mechanism $(q,p)$ is a \textsc{post-2} mechanism if there exist $\kappa_1,\kappa_2 \in [b,\beta]$ such that
$\kappa_1 \le \kappa_2$ and for every $v \in V$, we have
\begin{align*}
(q(v),p(v)) =
\begin{cases}
(0,0) & \textrm{if}~v_1 < \kappa_1 \\
\Big(1,b + \kappa_2 \big(1-\frac{b}{\kappa_1}\big)\Big) & \textrm{if}~v_2 > \kappa_2\\
\big(\frac{b}{\kappa_1},b\big) & \textrm{otherwise}
\end{cases}
\end{align*}
\end{defn}

\begin{figure}
  \centering
\includegraphics[width=3in]{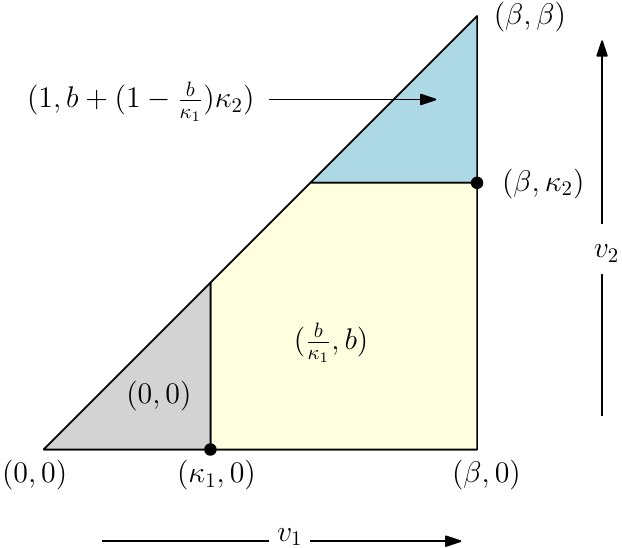}
\caption{A {\sc post-2} mechanism}
\label{fig:post2main}
\end{figure}
Figure \ref{fig:post2main} shows a {\sc post-2} mechanism.
A {\sc post-2} mechanism allows for randomization only when $v_2 \leq \kappa_2$ and $v_1 \ge \kappa_1$.
In this region the object is allocated with a {\sl fixed} probability $\frac{b}{\kappa_1}$.
This ensures that the range of a {\sc post-2} mechanism contains at most three outcomes.
The following lemma (proof is in Appendix \ref{sec:ppost2}) establishes that a {\sc post-2} mechanism is IC and IR.
\begin{lemma}
\label{lem:post2}
A {\sc post-2} mechanism is IC and IR.
\end{lemma}

The main result of the paper is the following.
\begin{theorem}
  \label{theo:main}
Suppose $x(1-F_1(x))$ is strictly concave. Then, there exists an optimal mechanism which is either a {\sc post-1} or a {\sc post-2} mechanism.
\end{theorem}

\noindent {\sc Proof Sketch.} The proof is quite involved and is presented in the Appendix \ref{sec:pmain}. We give a brief sketch here.
If we optimize over the class of mechanisms in which
all types pay less than $b$, then standard techniques lead to the optimality of a {\sc post-1} mechanism.
In the other case, we start with an arbitrary IC and IR mechanism in which some types pay more than $b$.
We divide the type space into three regions (see Figure \ref{fig:fp1} in the Appendix \ref{sec:pmain}) and construct
the restriction of the mechanism into each of these three regions. We provide upper bounds
on the revenue of the mechanism in each of these three regions. Finally, we put together
these upper bounds and show that a {\sc post-2} mechanism achieves this upper bound. The method of
coming up with this partition of the type space into three parts is somewhat involved. It is
described in detail in Appendix \ref{sec:pmain}.

Theorem \ref{theo:main}, and all subsequent results, will maintain the distributional assumption that
$x(1-F_1(x))$ is strictly concave. Notice that $x(1-F_1(x))$ is the expected revenue of the seller
to sell the object to the (unconstrained) agent (with marginal distribution $F_1$) at a posted price $x$.
This condition implies single-crossing of virtual value condition used in the single object auction
design literature. Since it is only a condition on the marginal distribution of
values of the agent, it allows for correlation among the values of the agent and the manager.

When the budget is high, the seller must expect that the agent takes decisions at most of the types.
In that case, a {\sc post-1} mechanism may be optimal. Similarly, if the budget is low, the decision-making
is given to the manager for more types. So, a {\sc post-2} mechanism may be optimal.
The next theorem shows that this intuition holds in general and provides a precise threshold on budget below which
{\sc post-2} is optimal and above which {\sc post-1} is optimal.

For each $i \in \{1,2\}$, let $\widetilde{\kappa}_i$ be defined as

\begin{align}\label{eq:kappai}
\widetilde{\kappa}_i := \arg \max_{x \in [0,\beta]} x(1-F_i(x))
\end{align}
This maximizer will be unique under appropriate conditions on $F_i$.
In particular, when $x(1-F_1(x))$ is strictly concave, $\widetilde{\kappa}_1$ is uniquely defined. We show that
$\widetilde{\kappa}_1$ is the threshold budget level, below which the optimal mechanism is a
{\sc post-2} mechanism and above which the optimal mechanism is a {\sc post-1} mechanism.

\begin{theorem}
\label{theo:post2}
Suppose $x(1-F_1(x))$ is strictly concave. Then, the following statements hold.
\begin{enumerate}
\item If $b \le \widetilde{\kappa}_1$, a {\sc post-2} mechanism is optimal.

\item If $b > \widetilde{\kappa}_1$, a {\sc post-1} mechanism is strictly optimal.
\end{enumerate}
\end{theorem}
\begin{proof}
{\sc Proof of (1)}.
Since $b \le \widetilde{\kappa}_1$ and $x(1-F_1(x))$ is strictly
concave, the optimal {\sc post-1} mechanism charges $b$ (since by definition a {\sc post-1} mechanism cannot charge more than $b$) and allocates the object to all types $v$ with $v_1 \ge b$.
This generates an expected revenue of $b(1-F_1(b))$. Now, consider the {\sc post-2} mechanism with $\kappa_1=\kappa_2=b$.
This generates an expected revenue of $b(1-F_1(b))$. Hence, the optimal {\sc post-2} mechanism is an optimal mechanism.

\noindent {\sc Proof of (2)}. Suppose $b > \widetilde{\kappa}_1$. Consider an arbitrary {\sc post-2}
mechanism with prices $\kappa_1,\kappa_2$. Notice that $b \le \kappa_1 \le \kappa_2$. Expected revenue from this mechanism is
\begin{align}\label{eq:r0}
b (1-F_1(\kappa_1)) + (1-\frac{b}{\kappa_1})\kappa_2 (1-F_2(\kappa_2)) &\le \frac{b}{\kappa_1} \kappa_1(1-F_1(\kappa_1)) + (1-\frac{b}{\kappa_1})\kappa_2 (1-F_1(\kappa_2))
\end{align}
where the inequality follows since $F_1$ first-order stochastic dominates $F_2$ (since $v_1 \ge v_2$).

If $\kappa_1=b$, the RHS of (\ref{eq:r0}) is $b(1-F_1(b))$, which is the revenue of a {\sc post-1} mechanism with price $b$.
But the optimal {\sc post-1} mechanism with a price $\widetilde{\kappa}_1$ generates a strictly greater revenue due to strict concavity of $x(1-F_1(x))$.

Now assume $\kappa_1 > b$. 
This implies
the RHS of (\ref{eq:r0}) is convex combination of $\kappa_1(1-F_1(\kappa_1))$ and $\kappa_2(1-F_1(\kappa_2))$.
Define
\begin{align*}
\widehat{\kappa} &:= \frac{b}{\kappa_1} \kappa_1 + (1-\frac{b}{\kappa_1})\kappa_2
\end{align*}
By strict concavity of $x(1-F_1(x))$, we get that the RHS of (\ref{eq:r0}) is less than
\begin{align}\label{eq:r00}
\widehat{\kappa} (1-F_1(\widehat{\kappa})) &\le \widetilde{\kappa}_1 (1-F_1(\widetilde{\kappa}_1))
\end{align}
where the inequality follows from the definition of $\widetilde{\kappa}_1$. Since $\widetilde{\kappa}_1 < b$,
the RHS of (\ref{eq:r00}) is the expected revenue of the optimal {\sc post-1} mechanism.
This implies that optimal {\sc post-1} mechanism generates strictly higher expected revenue than the {\sc post-2}
mechanism.
\end{proof}

\subsection{Description of optimal {\sc post-2} mechanisms}

Theorem \ref{theo:post2} shows that below a threshold budget, a {\sc post-1} mechanism is optimal, and above the threshold, a {\sc post-2} mechanism is optimal. While the optimal {\sc post-1} mechanism can be characterized by a simple closed-form expression (by setting posted price equal to  $\min \{\widetilde{\kappa}_1,b\}$), this is not the case for the optimal {\sc post-2} mechanism. The objective of this section is to provide more clarity on the description of the optimal {\sc post-2} mechanism. We will be able to better describe the optimal {\sc post-2} mechanism under additional conditions on the primitives of the model.

To remind, a {\sc post-2} mechanism
is described by two cutoffs, $\kappa_1,\kappa_2$ such that $b \le \kappa_1 \le \kappa_2 \le \beta$.
We partition the class of {\sc post-2} mechanisms as follows since they result in different characterizations.
A {\sc post-2} mechanism $(\kappa_1,\kappa_2)$ is
\begin{itemize}
\item a {\bf uniform} {\sc post-2} mechanism if $b \le \kappa_1=\kappa_2 \le \beta$;

\item an {\bf interior} {\sc post-2} mechanism if $b < \kappa_1 < \kappa_2 < \beta$.\footnote{In the interior {\sc post-2} mechanism definition, we preclude the cases of $\kappa_1 = b$ (identical to uniform {\sc post-2} with $\kappa_1 = b$) and $\kappa_2 = \beta$ (which is never optimal).}
\end{itemize}

Our first result in this section identifies the conditions under which either a uniform or an interior {\sc post-2} mechanism is optimal whenever a {\sc post-2} mechanism is optimal. 
The result also provides a closed-form description of the cutoffs that characterize the optimal mechanism.

We say the marginal distribution $F_2$ of the value of the manager satisfies {\bf single crossing (SC)}
if $x - \frac{1-F_2(x)}{f_2(x)}$ crosses zero exactly once. 
Since $xf_2(x)-(1-F_2(x))$ is the negative of the derivative of $x(1-F_2(x))$, the point at which $x-\frac{1-F_2(x)}{f_2(x)}$ equals zero is $\widetilde{\kappa}_2$ (defined in (\ref{eq:kappai})).
We say joint distribution $F$ satisfies {\bf monotone likelihood ratio property (MLRP)} if
$\frac{f_1(x)}{f_2(x)}$ is increasing in $x$. The MLRP property implies that $F_1$ first-order stochastic-dominates
$F_2$ (which comes for free in our model). 
This implies that $\widetilde{\kappa}_1 > \widetilde{\kappa}_2$, a fact we use later.
The proofs of following theorems are in Appendix \ref{sec:3.1proofs}.

\begin{theorem}
\label{theo:th4}
Suppose $F$ satisfies the MLRP property, $x(1-F_1(x))$ is strictly concave, and $F_2$ satisfies SC.
Then, the following are true if the optimal {\sc post-2} mechanism
generates strictly higher expected revenue than the optimal {\sc post-1} mechanism.

\begin{enumerate}

\item If $\frac{f_1(\widetilde{\kappa}_2)}{f_2(\widetilde{\kappa}_2)} \le 1$, then the optimal mechanism is a uniform {\sc post-2} mechanism. The cutoff of the optimal mechanism $\kappa^*$ is such that $\kappa^* \geq \widetilde{\kappa}_2$ and is a solution to 
\begin{align}\label{eq:part1}
xf_2(x) - (1-F_2(x)) &= b \Big( f_2(x) - f_1(x)\Big)
\end{align}

\item If $\frac{f_1(\widetilde{\kappa}_2)}{f_2(\widetilde{\kappa}_2)} > 1$, then the optimal mechanism is an interior {\sc post-2} mechanism. The optimal mechanism $(\kappa^*_1,\kappa^*_2)$ is such that
\begin{align}\label{eq:k2star}
\kappa^*_2 &= \widetilde{\kappa}_2
\end{align}
and $\kappa^*_1$ is the solution to
\begin{align}\label{eq:k1star}
x^2f_1(x) &= \kappa^*_2(1-F_2(\kappa^*_2))
\end{align}
\end{enumerate}
\end{theorem}

Theorem \ref{theo:th4} says that if the optimal {\sc post-2} mechanism is {\bf strictly} better than the optimal {\sc post-1} mechanism, then a description of optimal {\sc post-2} mechanism is possible under additional conditions on the primitives of the model. In particular, it provides two cases under which either a uniform {\sc post-2} or an interior {\sc post-2} mechanism is optimal with a precise expression to compute the posted price(s). 

Theorem \ref{theo:th4} assumes that the optimal {\sc post-2} mechanism generates {\sl strictly}
higher expected revenue than the optimal {\sc post-1} mechanism.
Note that Theorem \ref{theo:post2} only showed that if
$b \le \widetilde{\kappa}_1$, optimal {\sc post-2} mechanism generates {\sl weakly} higher payoff
than the optimal {\sc post-1} mechanism.
Therefore, it is unclear at what values of $b$ Theorem \ref{theo:th4} can be applied.

The following theorem addresses this partially by assuming only $b \le \widetilde{\kappa}_1$ (ensuring weak
optimality of {\sc post-2} mechanism due to Theorem \ref{theo:post2}). But for different ranges of $b$, it provides
sufficient conditions on distributions under which a uniform {\sc post-2}
mechanism is either strictly optimal or is equivalent to a {\sc post-1} mechanism.

\begin{theorem}
\label{theo:uniform}
Suppose $F$ satisfies the MLRP, $x(1-F_1(x))$ is strictly concave, $F_2$ satisfies SC, and $b \le \widetilde{\kappa}_1$.
Then, the following statements hold.

\begin{enumerate}

\item If $b < \widetilde{\kappa}_2$ and $\frac{f_1(\widetilde{\kappa}_2)}{f_2(\widetilde{\kappa}_2)} \le 1$,
a uniform {\sc post-2} mechanism with price $\kappa^* > b$ that solves the following equation is optimal:
\begin{align}\label{eq:upost2}
xf_2(x) - (1-F_2(x)) &= b \Big( f_2(x) - f_1(x)\Big)
\end{align}
\item If $\widetilde{\kappa}_2 \le b$ and $\frac{f_1(\widetilde{\kappa}_2)}{f_2(\widetilde{\kappa}_2)} > 1$,
a uniform {\sc post-2} mechanism with price $\kappa^*=b$ is optimal.

\item If $\widetilde{\kappa}_2 \le b$ and $\frac{f_1(\widetilde{\kappa}_2)}{f_2(\widetilde{\kappa}_2)} \le 1
< \frac{f_1(b)}{f_2(b)}$, a uniform {\sc post-2} mechanism with price $\kappa^*=b$ is optimal.
\end{enumerate}
\end{theorem}

The proof of Theorem \ref{theo:uniform} leverages the result in Theorem \ref{theo:th4} to derive condition under which a uniform {\sc post-2} mechanism is optimal and how the corresponding posted price can be computed. In part (1) of Theorem \ref{theo:uniform} we apply the first part of Theorem \ref{theo:th4} by additionally requiring that $b < \widetilde{\kappa}_2$.
This ensures that optimal {\sc post-2} mechanism generates strictly more revenue than optimal {\sc post-1} mechanism.
Under conditions (2) and (3) of Theorem \ref{theo:uniform} when $b \in [\widetilde{\kappa}_2,\widetilde{\kappa}_1]$, optimal {\sc post-1} and {\sc post-2} mechanisms are identical. Finally, when $b > \widetilde{\kappa}_1$, we have shown in Theorem \ref{theo:post2} that {\sc post-1} mechanism is strictly optimal.\\

\noindent {\sc Remark.} The sufficient conditions in Theorems \ref{theo:th4} and \ref{theo:uniform} can be satisfied for a variety of distributions. As an illustration,
consider a cdf $G$ with density $g$ on $[0,\beta]$. Let $x_1$ and $x_2$ be two draws from $[0,\beta]$
using $G$. We assign $v_1 := \max(x_1,x_2)$ and $v_2:=\min(x_1,x_2)$.
In this case $F_1(x)=[G(x)]^2$ and $F_2(x)=1-[1-G(x)]^2$.
We claim that {\sl if $xg(x)$ is increasing} then all the conditions imposed in
Theorem \ref{theo:uniform} (and all other theorems) on distributions are satisfied.
To see this, note that $f_1(x)=2g(x)G(x)$ and $f_2(x)=2g(x)(1-G(x))$.
Hence, $\frac{f_1(x)}{f_2(x)}=\frac{G(x)}{1-G(x)}$, which is strictly increasing.
Hence, MLRP holds. Next,
\begin{align*}
x-\frac{1-F_2(x)}{f_2(x)} &= x - \frac{[1-G(x)]^2}{2g(x)(1-G(x))} = x - \frac{1-G(x)}{2g(x)}
\end{align*}
Since $xg(x)$ is increasing, the above expression crosses zero exactly once. Finally, for strict concavity
of $x(1-F_1(x))$, we calculate its derivative: $-xf_1(x) + 1 - F_1(x) = -2xg(x)G(x) + 1- [G(x)]^2$. Since
$xg(x)$ is increasing, the derivative is strictly decreasing, establishing strict concavity of $x(1-G(x))$.
This shows that all the conditions of our theorems hold in this class of distributions (this includes the uniform distribution and some family of Beta distributions). Depending on
the exact nature of distribution (which determines the value of $\widetilde{\kappa}_i$ for each $i$),
and the budget,
we can be more precise about the optimal mechanism.

We expand this example to give implications of our theorems. Let $\beta=1$ and suppose $g$ is the uniform density: $g(x)=1$ for all $x \in [0,1]$. So, $G(x)=x$ and $F_1(x)=x^2, F_2(x)=1-(1-x)^2$. Note $xg(x)=x$ is strictly increasing, and hence, all conditions of our theorems hold. Now, we can compute $\widetilde{\kappa}_1$ and $\widetilde{\kappa}_2$ values as 
\begin{align*}
\widetilde{\kappa}_1 = \frac{1}{\sqrt{3}},~\widetilde{\kappa}_2=\frac{1}{3}
\end{align*}
Hence, by Theorem \ref{theo:post2}, if $b \le \frac{1}{\sqrt{3}}$ a {\sc post-2} mechanism is optimal and if $b > \frac{1}{\sqrt{3}}$ a {\sc post-1} mechanism with price $\frac{1}{\sqrt{3}}$ is strictly optimal.
Further, $f_1(x)=f_2(x)$ when $x=\frac{1}{2}$, and MLRP and the fact that $\widetilde{\kappa}_2=\frac{1}{3} < \frac{1}{2}$ implies $\frac{f_1(\widetilde{\kappa}_2)}{f_2(\widetilde{\kappa}_2)} < 1$. Hence, Theorem \ref{theo:uniform} implies: if $b \le \frac{1}{3}$, then a uniform {\sc post-2} mechanism with price greater than the budget (which solves equation (\ref{eq:upost2})) is optimal. If $b \in (\frac{1}{2},\frac{1}{\sqrt{3}}]$, then it is a uniform {\sc post-2} mechanism with price equal to the budget (Theorem \ref{theo:uniform} part (3) applies since $\frac{f_1(b)}{f_2(b)} > 1$). Theorem \ref{theo:uniform} is silent when $b \in (\frac{1}{3},\frac{1}{2}]$: it will be a {\sc post-2} mechanism (by Theorem \ref{theo:post2}), but can be an interior {\sc post-2} mechanism or a uniform {\sc post-2} mechanism.

\section{Relation to the literature}
\label{sec:lit}

Our paper is related to a couple of strands of literature in mechanism design. Before describing them,
we relate our work to two papers that seem most related to our work. The first is the work of \citet{Burkett16},
who studies a manager-agent model where the agent is participating in an auction mechanism with a third party.
In their model, the third party has proposed a mechanism for selling a single good. Given this third-party mechanism, the manager announces another mechanism (termed as {\em contract}) to the agent.
The sole purpose of the contract is to determine the amount the agent will bid in the third-party mechanism.
In their model, the agent's value of the good is the {\em only} private information - the manager's value can be determined from that of the agent's. The main result in this paper is that the optimal contract
for the manager is a ``budget-constraint'' contract. This optimal contract specifies a cap on the report of each type of the agent to
the third-party mechanism and involves no side-payments between the manager and the agent.~\footnote{In a related paper, \citet{Burkett15}
considers first-price and second-price auctions and compares their revenue and efficiency properties when a seller faces such manager-agent pairs.}

In our model, the values of the manager and the agent can be completely different (at a technical
level, \citet{Burkett16} has a one-dimensional mechanism design problem, whereas ours is a two-dimensional mechanism design problem).
Further, we do not model decision-making by our manager-agent pair via a contract. In other words, the sequential decision-making
in our model makes it quite different from \cite{Burkett15,Burkett16}.

A unique feature of our model is that the agent and the manager have a different value of the object to the firm.
The agent (who is better informed) does not persuade the manager to change his value. This is different from
\citet{MT19}, who study a model where a single good is sold to
a set of buyers, and each buyer is advised by a unique advisor with a bias. Before the start of the auction, there is a communication from
the advisor to the buyer, which influences how much the buyer bids in the auction.

\noindent {\sc Multidimensional mechanism design.} The type space of our agent is two-dimensional.
It is well known that the problem of finding an optimal mechanism for selling multiple goods
(even to a single buyer) is difficult. A long list of papers have shown the difficulties involved in extending the one-dimensional
results in \citet{Mussa78,Myerson81,Riley83} to multidimensional framework - see \citet{Armstrong00,Manelli07} as examples.
Even when the seller has just {\em two} objects, and there is just one buyer with additive valuations (i.e., value for both the objects
is the sum of values of both the objects), the optimal mechanism is difficult to describe~\citep{Manelli07,Dask17,Hart17}.
Indeed, strong conditions on priors are required to ensure that the optimal mechanism is deterministic~\citep{Pa11,BM22}.
This has inspired researchers to consider {\em approximately} optimal mechanisms~\citep{Chawla07,Chawla10,Hart17} or
additional robustness criteria for design~\citep{Carroll17}.
Compared to these problems, our two-dimensional mechanism design problem becomes tractable because of the nature
of incentive constraints, which in turn is a consequence of the preference of the agent.
\citet{GIYB21} study a model of single object
sale where the value of the object and an outside option are private information of the agent. They give
sufficient conditions under which a posted price mechanism is optimal. In our model, the manager's decision leads to
another option for the agent. However, this option is dependent on the mechanism of the seller. Further, the usual outside
option $(0,0)$ exists in our model.\\

\noindent {\sc Mechanism design with budget constraints.} In our model, the agent is budget-constrained, but the manager is not.
In the standard model, when there is a single object and the buyer(s) is budget constrained,
the space of mechanisms
is restricted to be such that payment is no more than the budget. This feasibility requirement on the mechanisms
essentially translates to violating the quasilinearity assumption of the buyer's preference for prices above the budget.
This introduces additional complications
for finding the optimal mechanism \citep{Laffont96,Che00,Pai14}. When the budget is private information,
the problem becomes even more complicated - see \citet{Che00} for a description of the optimal mechanism for the
single buyer case and \citet{Pai14} for a description of the optimal mechanism for the multiple buyers case.
All these mechanisms involve randomization, but the nature of randomization is quite different from ours.
This is because the source of randomization in all these papers is either due to budget being private information (hence, part of the type,
as in \citet{Che00,Pai14}) or because of multiple agents with budget being common knowledge (as in \citet{Laffont96,Pai14}).
Indeed, with a single agent and public budget, the optimal mechanism in a standard single object allocation model is a posted price mechanism.
This is in contrast with our result where we get a randomized optimal mechanism
even with the budget being common knowledge.
Also, the set of menus in the optimal mechanism in the standard single object auction with budget constraint may have more than three outcomes.
Further, the outcomes in the menu of these optimal mechanisms are not as simple as our {\sc post-2} mechanism.
Finally, like us, these papers assume that the budget is exogenously determined by the agent. If the buyer can choose his budget
constraint, then \citet{Baisa16} shows that the optimal mechanism in a multiple buyers setting allocates the object
efficiently whenever it is allocated - this is in contrast to the exogenous budget case \citep{Laffont96,Pai14}.
\citet{Li21} studies a model of financially constrained agents buying a single object when the mechanism designer
can inspect the budget at a cost (both value and budget are private information). The optimal mechanism is significantly
complicated and involves inspection.

\section{Discussions}
\label{sec:disc}

We have considered a model where a manager and an agent acquire an object
for the firm jointly. However, the manager has delegated the participation
in the mechanism to the agent with a budget constraint. The manager does not
question the agent as long as the agent does not violate the budget constraint.
However, the agent can come back to the manager if the payment exceeds the budget, in which
case the manager takes decisions. Besides the agent and the board of directors example
highlighted in the paper, there are other settings where such decision-making may be
plausible. Our results highlight the nuanced, but still simple, nature of the optimal mechanism
in such settings.

As we observed in the example in Section \ref{sec:illust}, the ex-post payoff of the manager in the optimal 
mechanism can be negative. 
However, the ex-ante payoff the manager in the optimal mechanism may be positive as the example below illustrates. We consider an example with $\beta = 1$ and uniform distribution with budget constraint $b = 0.25$. 
In this case, the optimal mechanism is a {\sc post-2} mechanism with $\kappa_1 = \kappa_2 = 0.39$.
The manager's payoff when agent chooses $(\frac{0.25}{0.39},0.25)$ is given by
\begin{align*}
\int_{0}^{0.39} 2 \big(y \frac{0.25}{0.39} - 0.25 \big) (1-0.39) dy = -0.059
\end{align*}
The manager's payoff when he chooses $(1,0.39)$ is given by
\begin{align*}
\int_{0.39}^{1} \int_{y}^{1} 2 \big(y - 0.39\big)dx dy = 0.075
\end{align*}
Therefore, the manager's ex-ante payoff from the mechanism is positive. Although our main interpretation is that the manager delegates to the agent due to decision-making costs, the example shows that the ex-ante payoff remains positive even without these costs.\footnote{For certain distributions, the ex-ante payoff of the manager may be negative too. In those cases, it is useful to think that the manager has a high decision-making cost.} The specifics of this delegation are beyond the scope of this paper (see \cite{Burkett16} for a delegation model with a budget constraint).

We make a few comments about possible extensions.
First, we have assumed that the value of the agent is more than that of the manager.
Under stronger conditions on distributions, we believe that our main results can be
extended when we relax this assumption. Second, we assume that the budget is
observed by the seller. Relaxing this assumption significantly complicates the
model. Some partial characterizations of optimal mechanisms seem possible
when the budget is also a private information.

\appendix

\section{Proof of Lemma \ref{lem:post2}}
\label{sec:ppost2}

\begin{proof}
Let $(q,p)$ be a {\sc post-2} mechanism defined by $\kappa_1,\kappa_2 \in [0,\beta]$ such that
$\kappa_1 \le \kappa_2$. If $\kappa_1=b$, the {\sc post-2} mechanism collapses to a {\sc post-1} mechanism, which
is IC and IR. Similarly, if $\kappa_2=\beta$, we have a modified {\sc post-1} mechanism where all types with $v_1 \ge \kappa_1$
are given the object with probability $\frac{b}{\kappa_1}$ at price $b$ and others do not get the object and pay zero.
Again, this is an IC and IR mechanism.

So, we consider the case $b < \kappa_1 \le \kappa_2 < \beta$.
In this case, the range of the {\sc post-2} mechanism contains
three outcomes: $(0,0), (1,b+ \kappa_2 (1-\frac{b}{\kappa_1})),$ and
$(\frac{b}{\kappa_1},b)$. Two of these outcomes are within the budget. Take a type $v \equiv (v_1,v_2)$. We consider three cases.
Denote $X:=R(q,p)$. \\

\noindent {\sc Case 1.} $v_2 < \kappa_2$. Then
 \begin{align}\label{eq:p1}
v_2 - b - \kappa_2 + b \frac{\kappa_2}{\kappa_1} &= (v_2-\kappa_2) \big(1-\frac{b}{\kappa_1}\big) + v_2 \frac{b}{\kappa_1} - b < v_2 \frac{b}{\kappa_1} - b
\end{align}
where the inequality follows from $b < \kappa_1$ and $v_2 < \kappa_2$. This implies $Ch(X;v_2) \subseteq X_b$ since $(1,b+ \kappa_2 (1-\frac{b}{\kappa_1})) \not\in Ch(X;v_2)$. Therefore, $Ch(X;v_2) \rhd_{v_1} Ch(X_b;v_1)$ does not hold. Hence, $Ch(X;v) = Ch(X_b;v_1)$  whenever $v_2 < \kappa_2$.

If $v_1 < \kappa_1$, then $Ch(X_b;v_1) = \{(0,0)\}$ since $v_1 \frac{b}{\kappa_1} - b < 0$. Hence, $(q(v),p(v)) = (0,0) \in Ch(X_b;v_1) = Ch(X;v)$.

If $v_1 \geq \kappa_1$, then $(\frac{b}{\kappa_1},b) \in Ch(X_b;v_1)$ since $v_1 \frac{b}{\kappa_1} - b \geq 0$. Hence, $(q(v),p(v)) = (\frac{b}{\kappa_1},b) \in Ch(X_b;v_1) = Ch(X;v)$.

\noindent {\sc Case 2.} $v_1 \ge v_2 \geq \kappa_2 \ge \kappa_1$ and $v_1 > \kappa_2$.
Since $v_i \geq \kappa_2$
\begin{align}
  \label{ineq:i1}
v_i - b - \kappa_2 + b \frac{\kappa_2}{\kappa_1} = (v_i-\kappa_2)(1-\frac{b}{\kappa_1}) + v_i \frac{b}{\kappa_1} - b \ge v_i \frac{b}{\kappa_1} - b \ge 0
\end{align}
where the first inequality follows from $v_i \ge \kappa_2$ and $\kappa_1 > b$, the second inequality follows from $v_i \ge \kappa_1$.

Hence, we get $(1,b + \kappa_2 (1-\frac{b}{\kappa_1})) \in Ch(X; v_i)$ in this case.
Since $v_1 > \kappa_2 \ge \kappa_1$ and $\kappa_1 > b$, both inequalities in (\ref{ineq:i1}) are strict for $v_i=v_1$. Hence, $Ch(X_b;v_1) = \{(\frac{b}{\kappa_1},b)\}$.
Therefore, $Ch(X;v_2) \rhd_{v_1} Ch(X_b;v_1)$ holds and $Ch(X;v) = Ch(X;v_2)$ in this case.

Now, when $v_2 > \kappa_2$, the first inequality is strict in (\ref{ineq:i1}) for $v_i=v_2$.
As a result, $Ch(X;v_2)=\{(1,b+\kappa_2 (1-\frac{b}{\kappa_1}))\}$.
This implies that $(q(v),p(v)) = (1,b + \kappa_2 (1-\frac{b}{\kappa_1})) \in Ch(X;v_2) = Ch(X; v)$.
When $v_2=\kappa_2$, the first inequality in (\ref{ineq:i1}) is an equality. Hence, $(q(v),p(v)) = (\frac{b}{\kappa_1},b) \in Ch(X;v_2)=Ch(X;v)$. \\

\noindent {\sc Case 3.} $v_1 = v_2 = \kappa_2$.
Then $Ch(X;v_2) \sim_{v_1} Ch(X_b;v_1)$ holds. Hence, $(q(v),p(v)) = (\frac{b}{\kappa_1},b) \in Ch(X_b;v_1) = Ch(X;v)$.
\end{proof}

\section{Proof of Theorem \ref{theo:main}}
\label{sec:pmain}

We begin the proof of this theorem with some preliminary lemmas.

\begin{lemma}
\label{lem:p2}
Suppose $v_i,v_i' \in [0,\beta]$ with $v_i' > v_i$. For any $X \subseteq Z$, if $(a',t') \in Ch(X;v_i')$ and
$(a,t) \in Ch(X;v_i)$, then $a \le a'$ and $t \le t'$.
\end{lemma}
\begin{proof}
Since $(a',t') \in Ch(X;v_i')$ and $(a,t) \in Ch(X;v_i)$, we have
\begin{align*}
a' v_i' - t' &\ge a v_i' - t \\
a v_i - t &\ge a' v_i - t'
\end{align*}
Adding gives $(a'-a)(v_i'-v_i) \ge 0$. Since $v'_i > v_i$, we have $a' \ge a$.
Then, the second inequality gives $(a'-a)v_i \le t'-t$. Since $a' \ge a$,
we get $t' \ge t$.
\end{proof}

\begin{lemma}
\label{lem:p3}
If $(q,p)$ is IC and IR, then there exists $\kappa_{(q,p)} \in (0,\beta]$ such that
for each $v_i \in [0,\beta]$ and each $(a,t) \in Ch(R(q,p);v_i)$, we have
\begin{align*}
v_i < \kappa_{(q,p)} \Longrightarrow t \le b \\
v_i > \kappa_{(q,p)} \Longrightarrow t > b
\end{align*}
\end{lemma}
\begin{proof}
By IR, if $(a,t) \in Ch(R(q,p);0)$ then $t \le 0$. Let $V^-:=\{v_i \in [0,\beta]:
t \le b~\forall~(a,t) \in Ch(R(q,p);v_i)\}$. Then, $0 \in V^-$.
Let $\kappa_{(q,p)} = \sup_{v_i \in V^-} v_i$. Note here that $\kappa_{(q,p)}$ does not depend on $i$. By Lemma \ref{lem:p2}, the result then follows.
\end{proof}

For any mechanism $(q,p)$, define the following set of types.
\begin{align*}
V^+(q,p) &= \{v: p(v) > b\}
\end{align*}
The proof of Theorem \ref{theo:main} consists of considering two classes of mechanisms, one where
$V^+(q,p)$ has zero Lebesgue measure and the other where it has non-zero Lebesgue measure. In particular,
consider the following partitioning of mechanisms.
\begin{align*}
M^- &= \{(q,p)~\textrm{is IC and IR}: V^+(q,p)~\textrm{has zero Lebesgue measure}\} \\
M^+ &= \{(q,p)~\textrm{is IC and IR}: (q,p) \notin M^-\}
\end{align*}

\subsection{\sc{post-2} is optimal in $M^+$}

\begin{lemma}
\label{lem:p7}
If $(q,p) \in M^+$ then for all $v$,
\begin{align*}
p(v) &\leq b \qquad \text{if}~v_2 < \kappa_{(q,p)}\\
p(v) &> b \qquad \text{if}~v_2 > \kappa_{(q,p)}
\end{align*}
\end{lemma}
\begin{proof}
Denote $X \equiv R(q,p)$. Lemma \ref{lem:p3} implies that for all $v$ with $v_2 < \kappa_{(q,p)}$ we have $Ch(X;v_2) \subseteq X_b$. Therefore, $Ch(X;v) = Ch(X_b;v_1)$. Incentive compatibility then implies that $p(v) \leq b$. Lemma \ref{lem:p3} also implies that for all $v$ with $v_2 > \kappa_{(q,p)}$ we have $Ch(X;v_2) \subseteq X\setminus X_b$. Pick any $(a,t) \in Ch(X;v_2)$ and $(a',t') \in X_b$ and note that $v_2 a - t > v_2 a' - t'$. Using $v_1 \geq v_2$ and $t > t'$ we derive
$v_1 a - t > v_1 a' - t'$.
Since this is true for all $(a,t) \in Ch(X;v_2)$ and $(a',t') \in X_b$, we conclude $Ch(X;v_2) \rhd_{v_1} Ch(X_b;v_1)$. Therefore, $Ch(X;v) = Ch(X;v_2) \subseteq X \setminus X_b$. Incentive compatibility then implies that $p(v) > b$.
\end{proof}

We now fix $(q,p)$ and show that there exists a {\sc post-2} mechanism that generates more
expected revenue than $(q,p)$. For simplicity, we will drop $(q,p)$ subscript from $\kappa_{(q,p)}$.
We will partition the type space $V = \{v \in [0,\beta]^2: v_1 \ge v_2\}$ into three parts
as (see Figure \ref{fig:fp1}):
\begin{align*}
V_{\ell} &= \{v \in V: v_1 \le \kappa\} \\
V_{h} &= \{v \in V: v_2 \ge \kappa\} \\
V_{m} &= V \setminus (V_1 \cup V_2)
\end{align*}
Note that $V_h \cap V_{\ell} = \{(\kappa,\kappa)\}$. Hence, $V_{\ell},V_h,V_m$ do not represent
an exact partitioning of $V$, but this will not influence our analysis.

\begin{figure}
\centering
\includegraphics[width=3in]{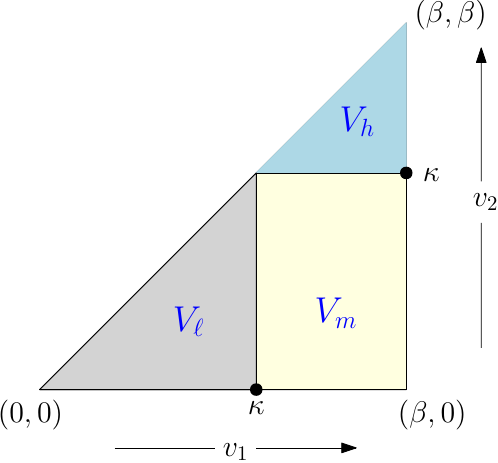}
\caption{Partitioning of type space $V$}
\label{fig:fp1}
\end{figure}

Now, consider the following type: $(\kappa,0)$ and define $q^\dag$ as the unique solution to
\begin{align}
\label{eq:qdag}
\kappa q^\dag - b &= \kappa q(\kappa,0) - p(\kappa,0)
\end{align}
Since $p(\kappa,0) \le b$, we get
\begin{align}
  \label{eq:qdag}
q^{\dag} &= q(\kappa,0) + \frac{b-p(\kappa,0)}{\kappa} \ge q(\kappa,0)
\end{align}
Now, consider the type $\kappa^{\epsilon} \equiv (\kappa+\epsilon,\kappa+\epsilon)$, where $\epsilon > 0$ but sufficiently small.
By Lemma \ref{lem:p7}, $p(\kappa^\epsilon) > b$.
Incentive compatibility of $(q,p)$ implies $(q(\kappa^\epsilon),p(\kappa^\epsilon)) \in Ch(R(q,p);\kappa+\epsilon)$. Hence,
\begin{align*}
(\kappa+\epsilon) q(\kappa^\epsilon) - p(\kappa^\epsilon) &\ge (\kappa+\epsilon) q(\kappa,0) - p(\kappa,0) = \epsilon q(\kappa,0) + \kappa q^\dag - b \\
\Longleftrightarrow q^\dag &\le q(\kappa^\epsilon) - \frac{p(\kappa^\epsilon) - b}{\kappa} + \frac{\epsilon}{\kappa}(q(\kappa^\epsilon) - q(\kappa,0)) \\
&\le q(\kappa^\epsilon) + \frac{\epsilon}{\kappa}(q(\kappa^\epsilon) - q(\kappa,0))
\end{align*}
where the last inequality uses $p(\kappa^\epsilon) > b$. As $\epsilon \rightarrow 0$, we get (using (\ref{eq:qdag}))
\begin{align}
\label{eq:qdag2}
q(\kappa,0) \le q^\dag &\le \lim \limits_{\epsilon \rightarrow 0} q(\kappa^\epsilon)
\end{align}
This shows that $q^\dag \in [0,1]$ is a valid allocation probability.
We also observe the following about utilities of types. Take any $v \equiv (v_1,v_2)$ with $v_2 > \kappa$.
By IC and the fact that $p(v) > b$, we get
\begin{align*}
u(v) &:= v_2 q(v) - p(v) \ge v_2 q(\kappa,0) - p(\kappa,0) \ge \kappa q(\kappa,0) - p(\kappa,0) = \kappa q^\dag - b
\end{align*}
where the last equality follows from the definition of $q^\dag$. Hence, for all $v$ with $v_2 > \kappa$, we have
\begin{align}
  \label{eq:udag}
u(v) &\ge \kappa q^\dag - b
\end{align}

Now, $q^\dag$ will be used to determine an upper bound on expected revenue of $(q,p)$. We will
show that this upper bound can be achieved by a {\sc post-2} mechanism.
We derive this upper bound by deriving upper bounds for expected revenue from $V_{\ell}, V_h,$ and $V_m$
separately.

\subsection{Upper bounding $V_{\ell}$}
\label{sec:vell}

For any $v \in V_{\ell}$ with $v_1 < \kappa$,  by Lemma \ref{lem:p3} we have
\begin{equation}\label{eq:vlIC}
(q(v),p(v)) \in Ch(R(q,p);v_1).
\end{equation}
Also, $(q(\kappa,0),p(\kappa,0)) \in Ch(R(q,p);\kappa)$. Hence, for any $v$ with $v_1 < \kappa$, IC constraints imply
\begin{align*}
v_1 q(v) - p(v) &\ge v_1 q(\kappa,0) - p(\kappa,0) \\
\kappa q(\kappa,0) - p(\kappa,0) &\ge \kappa q(v) - p(v)
\end{align*}
The first inequality together with equation (\ref{eq:qdag}) and the fact that $p(\kappa,0) \leq b$ implies for all $v$ with $v_1 < \kappa$,
\begin{align}\label{eq:qdag3}
v_1 q(v) - p(v) &\ge v_1 q^\dag - b
\end{align}
while the second inequality together with equation (\ref{eq:qdag}) implies for all $v$ with $v_1 < \kappa$,
\begin{align}\label{eq:qdag4}
\kappa q^\dag - b &\ge \kappa q(v) - p(v)
\end{align}

Restricting to $V_{\ell}$, we define a mechanism $(\tilde{q},\tilde{p}) : V_{\ell} \to [0,1] \times \mathbb{R}$ by setting $(\tilde{q}(v),\tilde{p}(v)) = (q^\dag,b)$ when $v_1 = \kappa$ and $(\tilde{q}(v),\tilde{p}(v)) = (q(v),p(v))$ otherwise. So the range of this mechanism is $R(\tilde{q},\tilde{p}) := \{ (q(v),p(v)): v_1 < \kappa\} \cup \{(q^\dag,b)\}$.
We say that the mechanism $(\tilde{q},\tilde{p})$ is IC in $V_{\ell}$ if 
\begin{equation*}
(\tilde{q}(v),\tilde{p}(v)) \in Ch(R(\tilde{q},\tilde{p});v_1) \qquad  \forall v \in V_{\ell}.
\end{equation*}

Since $\{(q(v),p(v)): v_1 < \kappa\} \subset R(q,p)$, expression (\ref{eq:vlIC}) together with definition of $(\tilde{q},\tilde{p})$ and inequality (\ref{eq:qdag3}) imply that $(\tilde{q}(v),\tilde{p}(v)) \in Ch(R(\tilde{q},\tilde{p});v_1)$ for all $v$ with $v_1 < \kappa$.
Inequality (\ref{eq:qdag4}) implies that $(\tilde{q}(v),\tilde{p}(v)) \in Ch(R(\tilde{q},\tilde{p});v_1)$ when $v_1 = \kappa$.
Therefore we conclude $(\tilde{q},\tilde{p})$ is IC in $V_{\ell}$. 
It is also IR by construction.

Define $M^{\ell}$ to be the set of all mechanisms defined on $V_{\ell} $ and that are IR and IC in $V_{\ell}$ such that
$(\tilde{q}(v),\tilde{p}(v)) = (q^\dag,b)$ for all $v$ with $v_1 = \kappa$.
The optimal mechanism in $M^{\ell}$ gives an upper bound on the expected revenue of $(q,p)$ in $V_{\ell}$.

We derive the optimal mechanism in $M^\ell$ in some steps. First, we show a lemma which further restricts
the class of mechanisms we need to consider inside $M^{\ell}$.
\begin{lemma}
\label{lem:p66}
For every $(\tilde{q},\tilde{p}) \in M^{\ell}$, there exists a mechanism $(\tilde{q}',\tilde{p}') \in M^\ell$ such that
\begin{align*}
\tilde{p}'(v_1,v_2) &= \tilde{p}'(v_1,v'_2)~\qquad~\forall~(v_1,v_2), (v_1,v'_2) \in V_{\ell} \\
\tilde{q}'(v_1,v_2) &= \tilde{q}'(v_1,v'_2)~\qquad~\forall~(v_1,v_2), (v_1,v'_2) \in V_{\ell}\\
\tilde{p}'(v) &\ge p(v)~\qquad~\forall~v \in V_{\ell}
\end{align*}
\end{lemma}
\begin{proof}
Pick any $(\tilde{q},\tilde{p}) \in M^\ell$.  Since $(\tilde{q},\tilde{p}) \in M^{\ell}$ is IC in $V_{\ell}$,  by definition we have $(\tilde{q}(v),\tilde{p}(v)) \in Ch(R(\tilde{q},\tilde{p});v_1)$ for all $v \in V_{\ell}$.
For every $v_1$ with $v \in V_{\ell}$, let
$Ch^*(R(\tilde{q},\tilde{p});v_1)$ be the set of outcomes defined as:
\begin{align*}
Ch^*(R(\tilde{q},\tilde{p});v_1) &:=
\{(a,t) \in Ch(R(\tilde{q},\tilde{p});v_1): t \ge t'~\forall~(a',t') \in Ch(R(\tilde{q},\tilde{p});v_1)\}
\end{align*}
By Claim \ref{cl:nempty} (in Appendix \ref{sec:rev}), $Ch^*(R(\tilde{q},\tilde{p});v_1)$ is non-empty. Note that if $v_1 \ne 0$, then
for any $(a,t), (a',t') \in Ch^*(R(\tilde{q},\tilde{p});v_1)$, we have $t=t'$ and $av_1 - t = a'v_1 - t'$. But $t=t'$
implies $a=a'$. Hence, $Ch^*(R(\tilde{q},\tilde{p});v_1)$ is a singleton if $v_1 \ne 0$.

Hence, we construct another mechanism $(\tilde{q}',\tilde{p}')$ such that
$(\tilde{q}'(v),\tilde{p}'(v))$ is assigned the unique element in $Ch^*(R(\tilde{q},\tilde{p});v_1)$ if
$v_1 \ne 0$. When $v_1=0$,
we choose $(a,t) \in Ch^*(R(\tilde{q},\tilde{p});0)$ such that $a$ is minimum across all outcomes
in $Ch^*(R(\tilde{q},\tilde{p});0)$ and set $(\tilde{q}'(v),\tilde{p}'(v))=(a,t)$.
Since $(\tilde{q},\tilde{p})$ is IC in $V_{\ell}$, $(\tilde{q}',\tilde{p}')$ is also IC in $V_{\ell}$.
Since $(\tilde{q},\tilde{p}) \in M^\ell$, $\tilde{q}(v) \le q^\dag$ and $\tilde{p}(v) \le b$
for all $v \in V^\ell$. Hence, by construction,
$\tilde{q}'(v) \le q^\dag$ and $\tilde{p}'(v) \le b$.
As a result, $(\tilde{q}',\tilde{p}') \in M^{\ell}$.
Further, by construction, $\tilde{p}'(v) \ge \tilde{p}(v)$ for all $v \in V^\ell$.
Also, by construction, $\tilde{p}'(v_1,v_2)=\tilde{p}'(v_1,v'_2)$ for all $(v_1,v_2),(v_1,v'_2) \in V^\ell$.
\end{proof}

Due to Lemma \ref{lem:p66}, we conclude that an optimal mechanism
in $M^\ell$ must belong to the following class of mechanisms:
\begin{align*}
M^{\ell\ell} &= \{(\tilde{q},\tilde{p}) \in M^\ell:
\tilde{p}(v_1,v_2) = \tilde{p}(v_1,v'_2), \tilde{q}(v_1,v_2) = \tilde{q}(v_1,v'_2)~\forall~(v_1,v_2),(v_1,v'_2) \in
V^\ell\}
\end{align*}

The expected revenue from any mechanism $(\tilde{q},\tilde{p}) \in M^{\ell \ell}$ is given by
\begin{align*}
\textsc{Rev}(\tilde{q},\tilde{p}) &= \int \limits_{V^{\ell}}\tilde{p}(v) f(v) dv = \int \limits_0^{\kappa} \tilde{p}(v_1,v_2)
\int \limits_0^{v_1}f(v_1,v_2)dv_2 dv_1 = \int \limits_0^{\kappa} \tilde{p}(v_1,v_2)f_1(v_1) dv_1
\end{align*}
where $f_1$ is the density of the marginal distribution of $v_1$. Since $(\tilde{q},\tilde{p}) \in M^{\ell \ell}$,
we let $\pi(v_1):=\tilde{p}(v_1,v_2)$ for all $(v_1,v_2) \in V^{\ell}$. Hence, the revenue
expression simplifies to
\begin{align*}
\textsc{Rev}(\tilde{q},\tilde{p}) &= \int \limits_0^{\kappa} \pi(v_1)f_1(v_1) dv_1
\end{align*}
Now, define for every $v_1 \in [0,\kappa]$, $\alpha(v_1):=\tilde{q}(v_1,v_2)$ and
\begin{align*}
\tilde{u}(v_1) &:= \alpha(v_1) v_1 - \pi(v_1)
\end{align*}
Since the IC constraints involve type of agent ($v_1$), standard envelope theorem
arguments imply for all $v_1 \in [0,\kappa]$,
\begin{align*}
\tilde{u}(v_1) &= \tilde{u}(0) + \int \limits_0^{v_1} \alpha(x)dx
\end{align*}
and $\frac{d\tilde{u}}{dx} = \alpha(x)$ for almost all $x \in [0,\kappa]$.
Using this and rewriting the revenue expression, we get
\begin{align}
  \textsc{Rev}(\tilde{q},\tilde{p}) &= \int \limits_0^{\kappa} \pi(v_1)f_1(v_1) dv_1 \nonumber \\
  &= \int \limits_0^{\kappa} \big[ \frac{d\tilde{u}}{dv_1} v_1 - \tilde{u}(v_1) \big]f_1(v_1) dv_1 \nonumber \\
  &= \kappa \tilde{u}(\kappa) f_1(\kappa) - \int \limits_0^{\kappa}\Big(v_1 \frac{df_1}{dv_1} + 2f_1(v_1)\Big) \tilde{u}(v_1)dv_1
  \label{eq:el1}
\end{align}

We now define another mechanism $(q^*,p^*)$ from $(\tilde{q},\tilde{p})$.
For every $v_1 \in [0,\kappa]$, we set $\alpha^*(v_1):=\alpha(\kappa)$ and
$\pi^*(v_1):=\pi(\kappa)$ if $\alpha(\kappa)v_1 - \pi(\kappa) \ge 0$
and assign $(\alpha^*(v_1),\pi^*(v_1)):= (0,0)$ otherwise.
Extend this to entire $V^{\ell}$ in the usual way using Lemma \ref{lem:p66}.

Clearly, this
defines an IC mechanism in $M^{\ell \ell}$. So, for
$v_1 < \frac{\pi(\kappa)}{\alpha(\kappa)}=\frac{b}{q^\dag}$, we get $u^*(v_1)=0$.
If $v_1 \ge \frac{b}{q^\dag}$, we have
$u^*(v_1)=v_1 \alpha(\kappa) - \pi(\kappa) \le \tilde{u}(v_1)$,
where the inequality follows from IC of $(\tilde{q},\tilde{p})$.
As a result, $u^*(v_1) \le \tilde{u}(v_1)$ for all $v_1 \le \kappa$ with equality
holding for $v_1=\kappa$.
Since $v_1(1-F_1(v_1))$ is concave, we know that
$v_1 \frac{df_1}{dv_1} + 2f_1(v_1) \ge 0$ for all $v_1$.
Using the expression in (\ref{eq:el1}), we conclude that
$\textsc{Rev}(q^*,p^*) \ge \textsc{Rev}(\tilde{q},\tilde{p})$.

\begin{figure}[!hbt]
\centering
\includegraphics[width=3in]{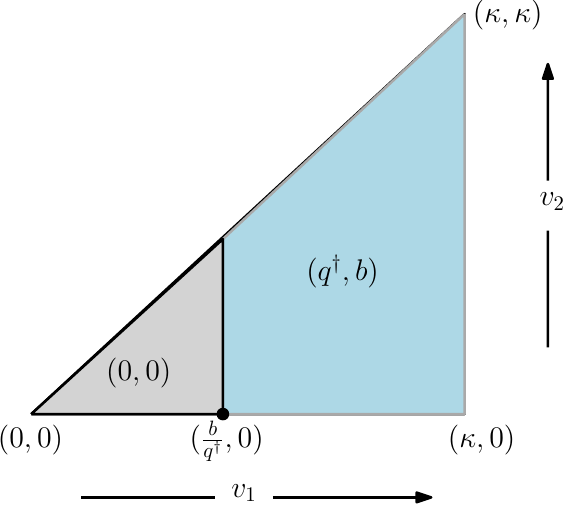}
\caption{Optimal mechanism in $M^{\ell}$ for type space $V_{\ell}$}
\label{fig:post2new}
\end{figure}

Since $(\alpha(\kappa),\pi(\kappa)) = (q^\dag,b)$ by construction, an optimal mechanism in $M^{\ell \ell}$ is described
by a cutoff $\frac{b}{q^\dag}$ such that for all types with $v_1 \in [\frac{b}{q^\dag},\kappa]$,
the object is allocated with probability $q^\dag$ and payment
$b$. The object is not allocated and payment is zero for all other types.
This is shown in Figure \ref{fig:post2new}.

\subsection{Upper bounding $V_h$}
\label{sec:veh}

Let $V^-_h:=V_h \setminus \{(v_1,\kappa):v_1 \ge \kappa\}$.
Consider a mechanism $(\hat{q},\hat{p}) : V^-_h \to [0,1] \times \mathbb{R}$ such that $(\hat{q}(v),\hat{p}(v)) = (q(v),p(v))$ for all $v \in V_h^-$. We say that the mechanism $(\hat{q},\hat{p})$ is IC in $V_h^-$ if 
\begin{equation*}
(\hat{q}(v),\hat{p}(v)) \in Ch(R(\hat{q},\hat{p});v_2)
\end{equation*}
For all $v \in V_h^-$, by definition of $V_h$, we have $p(v) > b$.
This also implies that $(q(v),p(v)) \in Ch(R(q,p);v_2)$. But since $R(\hat{q},\hat{p}) \subseteq R(q,p)$ we conclude that $(\hat{q},\hat{p})$ is IC in $V_h^-$.
This implies that for any $v \equiv (v_1,v_2), v' \equiv (v'_1,v'_2) \in V^-_h$,
with $v'_2 > v_2$, we must have
\begin{align*}
v'_2\hat{q}(v') - \hat{p}(v') &\ge v'_2 \hat{q}(v) - \hat{p}(v) \\
v_2 \hat{q}(v) - \hat{p}(v) &\ge v_2 \hat{q}(v') - \hat{p}(v')
\end{align*}
Adding them gives $\hat{q}(v') \ge \hat{q}(v)$. Using (\ref{eq:qdag2}), we get $\hat{q}(v) \ge q^\dag$ for all $v \in V^-_h$.
Further, inequality (\ref{eq:udag}) implies that for any $v \in V^-_h$, we have $v_2 \hat{q}(v) - \hat{p}(v) \ge kq^\dag - b$.

We can then extend the restriction of $(\hat{q},\hat{p})$ to $V^-_h$ to the entire $V_h$ by taking convergent sequences of
$\hat{q}(v),\hat{p}(v)$ in $V^-_h$. This will define a new mechanism $(\tilde{q},\tilde{p})$ which is IC in $V_h$ and satisfies three properties:
(i) $\tilde{p}(v) \ge b$, (ii) $\tilde{q}(v) \ge q^\dag$ and (iii)
$\tilde{u}(v):=v_2 \tilde{q}(v) - \tilde{p}(v) \ge \kappa q^\dag - b$ for all $v \in V_h$.

Hence, define $M^h$ to be the set of all mechanisms that are IC in $V_h$ such that each
$(\tilde{q},\tilde{p}) \in M^h$ satisfies $\tilde{q}(v) \ge q^\dag$ and
$\tilde{u}(v) \ge \kappa q^\dag - b$  for all $v \in V_h$.\footnote{We have ignored the constraint
$\tilde{p}(v) \ge b$ and hence, considered a larger set of mechanisms.}
The optimal mechanism in $M^h$ gives an upper bound on the expected revenue of $(q,p)$ in $V_h$.

Since the IC constraints are determined by $v_2$ for any $(v_1,v_2) \in V_h$, analogous to
Lemma \ref{lem:p66}, we can assume that for every mechanism $(\tilde{q},\tilde{p}) \in M^h$,
\begin{align*}
\tilde{q}(v_1,v_2) = \tilde{q}(v'_1,v_2) ~\qquad~\forall~(v_1,v_2),(v'_1,v_2) \in V_h \\
\tilde{p}(v_1,v_2) = \tilde{p}(v'_1,v_2) ~\qquad~\forall~(v_1,v_2),(v'_1,v_2) \in V_h
\end{align*}

Hence, we can set $\alpha(v_2)=\tilde{q}(v_1,v_2), \pi(v_2)=\tilde{p}(v_2)$ and
$u(v_2)=\tilde{u}(v_1,v_2)$ for all $v_2 \in [\kappa,\beta]$. IC constraints define a
one dimensional mechanism design problem of allocating a single object in the type
space $[\kappa,\beta]$ with IR constraints replaced by $u(v_2) \ge \kappa q^\dag - b$.
Further, the allocation probabilities lie in $[q^\dag,1]$. We know from \cite{Manelli07}
and \cite{borgers15} that the extreme points of such optimization problems assign
extreme allocation probabilities to each type, i.e., object is either given with probability
$q^\dag$ or $1$ at every type and the lowest type is assigned utility equal to $\kappa q^\dag - b$.

\begin{figure}[!hbt]
\centering
\includegraphics[width=3in]{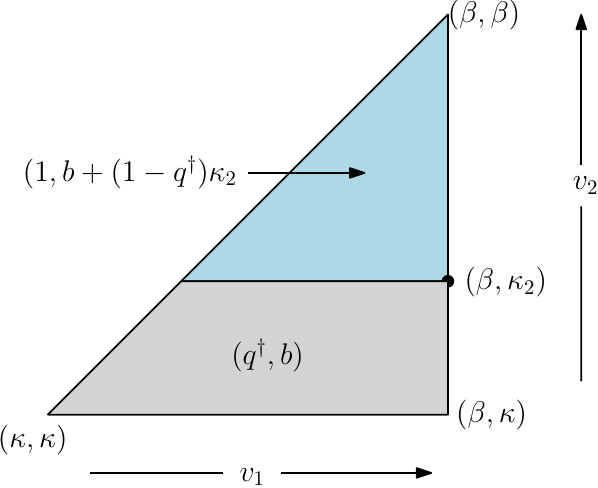}
\caption{Optimal mechanism in $M^{h}$ for type space $V_h$}
\label{fig:post1new}
\end{figure}

Hence, the expected revenue maximizing mechanism $(q^*,p^*)$ in $M^h$ must involve a cutoff $\kappa_2 \in [\kappa,\beta]$,
such that for all $(v_1,v_2) \in V_h$ with $v_2 \ge \kappa_2$, we have $q^*(v_1,v_2) = 1$
and for all $(v_1,v_2) \in V_h$ with $v_2 < \kappa_2$, we have $q^*(v_1,v_2) = q^\dag$.
Using the revenue equivalence (envelope theorem) arguments and the fact that
lowest possible utility is $\kappa q^\dag - b$, we conclude that
$p^*(v)=b$ for all $v \in V_h$ with $v_2 < \kappa_2$
and $p^*(v) = b + (1-q^\dag)\kappa_2$ for all $v \in V_h$ with $v_2 \ge \kappa_2$.
This is shown in Figure \ref{fig:post1new}.

The expected revenue from such a mechanism is
\begin{align}\label{eq:revuh}
b (1-F_2(\kappa)) + (1-F_2(\kappa_2))(1-q^\dag)\kappa_2
\end{align}
where $F_2$ is the marginal cdf of $v_2$.
Hence, an upper bound on the expected revenue of $(q,p)$ from the region $V_h$ is given by
a $\kappa^*_2$ such that an allocation probability $q^\dag$ at price $b$ is given to all
types in $V_h$, but the additional allocation probability $(1-q^\dag)$ is given to types
$v$ with $v_2 \ge \kappa^*_2$ at price $(1-q^\dag)\kappa^*_2$. The value of $\kappa^*_2$
solves
\begin{align}\label{eq:kappa2opt}
\max_{\kappa_2 \in [\kappa,\beta]} \kappa_2 (1-F_2(\kappa_2))
\end{align}

\subsection{Upper bounding $V_m$ and proof of Theorem \ref{theo:main}}

Now, consider the following {\sl post-2} mechanism $(q^*,p^*)$:
\begin{align*}
(q^*(v),p^*(v)) =
\begin{cases}
(0,0) & \textrm{if}~v_1 < \frac{b}{q^\dag} \\
\Big(1,b + \kappa^*_2 \big(1-q^\dag\big)\Big) & \textrm{if}~v_2 > \kappa^*_2\\
\big(q^\dag,b\big) & \textrm{otherwise}
\end{cases}
\end{align*}
where $\kappa^*_2$ is as defined in (\ref{eq:kappa2opt}). This is shown in Figure \ref{fig:post2}.

\begin{figure}[!hbt]
\centering
\includegraphics[width=3in]{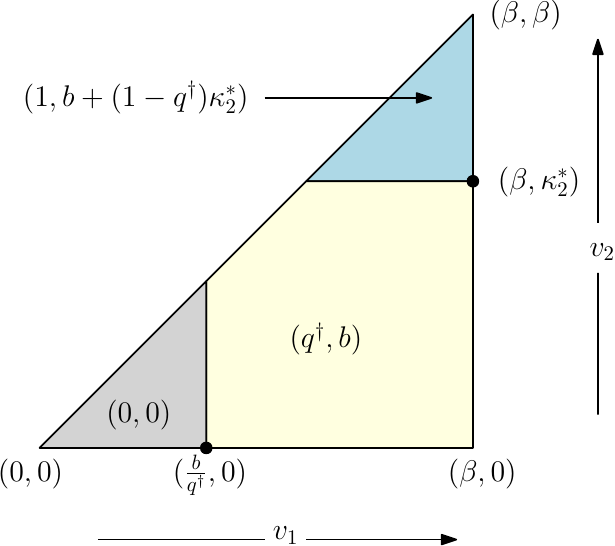}
\caption{A {\sc post-2} mechanism that revenue dominates $(q,p)$}
\label{fig:post2}
\end{figure}

Clearly, the {\sc post-2}
mechanism $(q^*,p^*)$ coincides with the mechanisms achieving the upper bound
in regions $V_{\ell}$ and $V_h$. The types in the $V_m$ region do not pay more
than $b$ in the mechanism $(q,p)$. Hence, this {\sc post-2} mechanism generates
greater expected revenue than $(q,p)$. This concludes the proof.

\subsection{\sc{post-1} is optimal in $M^-$}

\begin{lemma}
\label{lem:p4}
For any type $v \equiv (v_1,v_2)$ and any $X \subseteq Z$,
if $(a',t') \in Ch(X;v_2)$ and $t' > b$, then for any $(a,t) \in X$ with $t \le b$,
we have $v_1 a' - t' \ge v_1 a - t$ and the inequality is strict if $v_1 > v_2$.
\end{lemma}
\begin{proof}
Since $(a',t') \in Ch(X;v_2)$ and $t' > b \ge t$ we have $(a'-a) v_2 \ge t' - t > 0$.
This implies $a' > a$. Since $v_1 \ge v_2$, we get $(a'-a)v_1 \ge t'-t$ with strict inequality
holding if $v_1 > v_2$.
\end{proof}

\begin{lemma}
  \label{lem:p5}
If $(q,p) \in M^-$, then there exists another
$(q',p') \in M^-$ such that for every $v$ we have
$p'(v) \le b$ for all $v$ and $(q(v),p(v)) = (q'(v),p'(v))$
for all $v \ne (\beta,\beta)$.
\end{lemma}
\begin{proof}
Let $R(q,p)=X$. For any $v \equiv (v_1,v_2)$ with $v_2 < \beta$, suppose $p(v) > b$.
By Lemma \ref{lem:p3}, for all $v'_2 > v_2$, for any $(a',t') \in Ch(X;v'_2)$ we have $t > b$.

Take any $v' \equiv (v'_1,v'_2) \in V$ and suppose
$(a,t) \in Ch(X_b;v'_1)$. By Lemma \ref{lem:p4}, if $v'_1 > v'_2$,
we have $v'_1 a' - t' > v'_1 a - t$. Hence, for all $v' \equiv (v'_1,v'_2)$
such that $v'_1 > v'_2 > v_2$, we have $Ch(X;v') = Ch(X;v_2)$.
Hence, $p(v') > b$. Since $v_2 < \beta$, this shows that the set of
$v'$ with $p(v') > b$ has positive Lebesgue measure.
This contradicts the fact that $(q,p) \in M^-$.

Hence, we conclude that for any $v \equiv (v_1,v_2)$ with
$p(v) > b$, we must have $v_1=v_2=\beta$. Then, consider
$X':=X \setminus \{q(\beta,\beta),p(\beta,\beta)\}$. Denote by
$\bar{X}'$ the closure of $X'$. Pick some $(a,t) \in Ch(\bar{X}',\beta)$
and define $q'(\beta,\beta):=a$ and $p'(\beta,\beta)=t$.
By definition, $t \le b$. Further, let $(q'(v),p'(v)) = (q(v),p(v))$
for all $v \ne (\beta,\beta)$.
It is easily verified that $(q',p')$ is an IC and IR mechanism (since $(q,p)$ is IC and IR).
Further, $p'(v) \le b$ for all $v$.
\end{proof}

Hence, we can assume that the mechanism $(q,p)$ satisfies $p(v) \le b$ for all $v$.
By mimicking the arguments in Section \ref{sec:vell}, we can conclude that
the expected revenue of $(q,p)$ is dominated by an IC and IR mechanism $(\tilde{q},\tilde{p})$ such that for some $\kappa_1 \in [0,\beta]$,
we have $\tilde{q}(v)=q(\beta,0)$ and $\tilde{p}(v)=p(\beta,0)=\kappa_1 q(\beta,0)$ for all $v$ with $v_1 \ge \kappa_1$.
And $(\tilde{q}(v),\tilde{p}(v)) \equiv (0,0)$ for all other $v$.

So, every optimal mechanism in $M^-$ is characterized
by a posted-price $\kappa_1$ and an allocation probability $\alpha$, where the object
is allocated with this probability at price $\alpha \kappa_1$ if the value of the agent is above
the posted price. The optimal such posted price is the solution to the following program:
\begin{align*}
\max_{\kappa_1,\alpha} \kappa_1 \alpha (1-F_1(\kappa_1)) \\
\textrm{subject to}~\alpha \kappa_1 \le b,~~\alpha \in [0,1]
\end{align*}

We argue that the optimal solution to this program must have $\alpha=1$. To see this, let $\kappa^*_1$ be the unique solution to the following
optimization
$$\max_{\kappa_1 \in [0,b]} \kappa_1 (1-F_1(\kappa_1)).$$
The fact that this optimization program has a unique solution follows from the fact that $x - xF_1(x)$ is strictly concave. Hence, the revenue from the solution when $\alpha=1$ is $\kappa^*_1(1-F_1(\kappa^*_1))$.
Now, suppose the optimal solution has $\hat{\kappa}_1$ and $\hat{\alpha}_1$. Note that the $\hat{\kappa}_1\hat{\alpha} \le b$. So, define
$\tilde{\kappa}_1=\hat{\kappa}_1\hat{\alpha} \le b$. By definition,
\begin{align*}
\kappa^*_1(1-F_1(\kappa^*_1)) &\ge \tilde{\kappa}_1 (1-F_1(\tilde{\kappa}_1)) \\
&= \hat{\kappa}_1\hat{\alpha} (1-F_1(\hat{\kappa}_1\hat{\alpha})) \\
&\ge \hat{\kappa}_1\hat{\alpha} (1-F_1(\hat{\kappa}_1)),
\end{align*}
where the final inequality used the fact that $F_1(\hat{\kappa}_1\hat{\alpha}) \le F_1(\hat{\kappa}_1)$. This implies that the optimal solution
must have $\alpha=1$ and $\kappa_1$ must be the unique solution to $\kappa_1(1-F_1(\kappa_1))$ with the constraint $\kappa_1 \in [0,b]$. Hence, the optimal solution in $M^-$
must be a {\sc post-1} mechanism, where the posted price is a unique solution to
$$\max_{\kappa_1 \in [0,b]}\kappa_1(1-F_1(\kappa_1)).$$

\section{Proofs of Theorems \ref{theo:th4} and \ref{theo:uniform}}
\label{sec:3.1proofs}
In this appendix, we prove the theorems that describe optimal {\sc post-2} mechanism.

\subsection{Proof of Theorem \ref{theo:th4}}
Consider the optimal {\sc post-2} mechanism given by $(\kappa^*_1, \kappa^*_2)$. If $b = \kappa^*_1$, the expected revenue of the optimal {\sc post-2} mechanism is $b(1-F_1(b))$, which is also the expected revenue from the {\sc post-1} mechanism with price $b$. Since the optimal {\sc post-2} mechanism
generates strictly higher expected revenue than the optimal {\sc post-1} mechanism,
we conclude that $b < \kappa^*_1$.  A similar argument implies that $\kappa^*_2 < \beta$. Therefore, Theorem \ref{theo:main} and conditions of this theorem imply that either a uniform {\sc post-2} mechanism with $b < \kappa^*_1 = \kappa^*_2 < \beta$
or an interior {\sc post-2} mechanism is optimal.

We now prove three claims before proceeding with the proof.

\begin{claim}
\label{cl:cc1}
If an optimal {\sc post-2} mechanism is a uniform {\sc post-2} mechanism with $\kappa_1=\kappa_2=\kappa^*$, then
$\kappa^* \ge \widetilde{\kappa}_2$.
\end{claim}
\begin{proof}[of Claim~\ref{cl:cc1}]
Assume for contradiction $\kappa^* < \widetilde{\kappa}_2$. Then, the expected revenue from this mechanism is
\begin{align*}
b (1-F_1(\kappa^*)) + (1-\frac{b}{\kappa^*})\kappa^*(1-F_2(\kappa^*)) &< b (1-F_1(\kappa^*)) + (1-\frac{b}{\kappa^*})\widetilde{\kappa}_2(1-F_2(\widetilde{\kappa}_2)),
\end{align*}
where the inequality follows from the definition of $\widetilde{\kappa}_2$ and $F_2$ satisfying SC.
But the RHS is the expected revenue from the interior {\sc post-2} mechanism $(\kappa_1=\kappa^*, \kappa_2=\widetilde{\kappa}_2)$.
This contradicts the optimality of the uniform {\sc post-2} mechanism.
\end{proof}

\begin{claim}
  \label{cl:cc2}
If an optimal {\sc post-2} mechanism is a uniform {\sc post-2} mechanism with $\kappa_1=\kappa_2=\kappa^*$, then
$\kappa^*$ is the solution to
\begin{align*}
xf_2(x) - (1-F_2(x)) &= b \Big( f_2(x) - f_1(x) \Big),
\end{align*}
and $\frac{f_1(\kappa^*)}{f_2(\kappa^*)} \le 1$.
\end{claim}
\begin{proof}[of Claim~\ref{cl:cc2}]
Since $\kappa^* \in (b,\beta)$, it must be a solution
to the first order condition. The expected revenue from an arbitrary uniform {\sc post-2}
mechanism $\kappa=\kappa_1=\kappa_2$ is
\begin{align*}
b(1-F_1(\kappa))+ (\kappa-b) (1-F_2(\kappa))
\end{align*}
The first order condition at $\kappa^*$ is
\begin{align}
-bf_1(\kappa^*) + (1-F_2(\kappa^*)) - (\kappa^*-b)f_2(\kappa^*) &= 0 \nonumber \\
\Longleftrightarrow \kappa^* - \frac{1-F_2(\kappa^*)}{f_2(\kappa^*)} &= b \Big( 1 - \frac{f_1(\kappa^*)}{f_2(\kappa^*)}\Big) \label{eq:k11}
\end{align}
Using Claim \ref{cl:cc1}, $\kappa^* \ge \widetilde{\kappa}_2$. Since $F_2$ satisfies
SC, LHS of (\ref{eq:k11}) is non-negative. Hence, $\frac{f_1(\kappa^*)}{f_2(\kappa^*)} \le 1$.
Hence, the value of
$\kappa^*$ is the solution to the first order condition (\ref{eq:k11}).
\end{proof}

\begin{claim}
\label{cl:interior}
If the optimal {\sc post-2} mechanism $(\kappa^*_1,\kappa^*_2)$ is an interior {\sc post-2} mechanism, then
\begin{align*}
\kappa^*_2 &= \widetilde{\kappa}_2
\end{align*}
and $\kappa^*_1$ is the solution to
\begin{align*}
x^2f_1(x) &= \kappa^*_2(1-F_2(\kappa^*_2))
\end{align*}
\end{claim}

\begin{proof}[of Claim~\ref{cl:interior}]
The expected revenue generated
from an arbitrary {\sc post-2} mechanism $(\kappa_1,\kappa_2)$ is given by
\begin{align}\label{eq:p2rev}
b (1-F_1(\kappa_1)) + (1-\frac{b}{\kappa_1})\kappa_2 (1-F_2(\kappa_2)) &= b (1-F_1(\kappa_1)) + (1-\frac{b}{\kappa_1})\textsc{Rev}(\kappa_2),
\end{align}
where $\textsc{Rev}(\kappa_2):=\kappa_2 (1-F_2(\kappa_2))$.
So, $(\kappa^*_1,\kappa^*_2)$ maximizes expression (\ref{eq:p2rev}) under the
constraint that $b \le \kappa_1 \le \kappa_2 \le \beta$.

Fixing $\kappa^*_1$, by (\ref{eq:p2rev}), $\kappa^*_2$ must be a solution to
$\max_{x \in [\kappa^*_1,\beta]}x(1-F_2(x))$. Since $F_2$ satisfies SC, then
the expression $x(1-F_2(x))$ is increasing till $\widetilde{\kappa}_2$ and
decreasing after that.
Hence, either $\kappa^*_2=\kappa^*_1$ or $\kappa^*_2=\widetilde{\kappa}_2$.
Since $\kappa^*_1 < \kappa^*_2 < \beta$, we conclude that $\kappa^*_2 = \widetilde{\kappa}_2$.

We now denote $\textsc{Rev}_2:= \kappa^*_2(1-F_2(\kappa^*_2))$. Hence, the expected revenue
of any {\sc post-2} mechanism $(\kappa_1,\kappa^*_2)$ is given by
\begin{align}\label{eq:p2rev2}
b (1-F_1(\kappa_1)) + (1-\frac{b}{\kappa_1})\textsc{Rev}_2
\end{align}

Since $b < \kappa_1^* < \kappa^*_2$, we conclude that $\kappa_1^*$ must satisfy first order
conditions of the expression in (\ref{eq:p2rev}). Differentiating with respect to $\kappa_1$,
the expression in (\ref{eq:p2rev2}), we get
\begin{align}\label{eq:p2rev3}
-bf_1(\kappa_1) + \frac{b}{(\kappa_1)^2} \textsc{Rev}_2 = \frac{b}{(\kappa_1)^2}\Big(\textsc{Rev}_2 - (\kappa_1)^2f_1(\kappa_1)\Big)
\end{align}
We argue that $x^2f_1(x)$ is strictly increasing. This is because its derivative $2xf_1(x)+x^2\frac{df_1}{dx}=x(2f_1(x)+x\frac{df_1}{dx}) > 0$
by strict concavity of $x(1-F_1(x))$. Hence, expression (\ref{eq:p2rev3}) is positive
till $\kappa_1$ is such that $\textsc{Rev}_2=(\kappa_1)^2 f_1(\kappa_1)$ and negative
after that. This means there is a unique optimum to the revenue expression in (\ref{eq:p2rev2}),
which is given by the solution to $x^2 f_1(x)=\textsc{Rev}_2$.
\end{proof}

Now, we proceed to the proofs of parts (1) and (2). \\

\noindent {\sc Proof of Part (1).} Suppose $\frac{f_1(\widetilde{\kappa}_2)}{f_2(\widetilde{\kappa}_2)} \le 1$. Assume
for contradiction that an interior {\sc post-2} mechanism
$(\kappa^*_1,\kappa^*_2)$ is an optimal mechanism. By Claim \ref{cl:interior},
$\kappa^*_2=\widetilde{\kappa}_2$ and
\begin{align*}
(\kappa^*_1)^2 f_1(\kappa^*_1) &= \widetilde{\kappa}_2 (1-F_2(\widetilde{\kappa}_2)) = (\widetilde{\kappa}_2)^2 f_2(\widetilde{\kappa}_2),
\end{align*}
where the second equality comes from the optimality of $\widetilde{\kappa}_2$ to $\max x (1-F_2(x))$.
Since $x^2 f_1(x)$ is strictly increasing and $\widetilde{\kappa}_2 > \kappa^*_1$, we conclude that
\begin{align*}
(\widetilde{\kappa}_2)^2 f_1(\widetilde{\kappa}_2) > (\widetilde{\kappa}_2)^2 f_2(\widetilde{\kappa}_2)
\end{align*}
As a result, we get
\begin{align*}
\frac{f_1(\widetilde{\kappa}_2)}{f_2(\widetilde{\kappa}_2)} > 1
\end{align*}
which is a contradiction to our assumption. Hence, a uniform {\sc post-2} mechanism is an optimal mechanism. Using Claim \ref{cl:cc2} proves equation (\ref{eq:part1}).\\

\noindent {\sc Proof of Part (2).} Assume for contradiction that a uniform {\sc post-2} mechanism is optimal. Since $\frac{f_1(\widetilde{\kappa}_2)}{f_2(\widetilde{\kappa}_2)} > 1$, by MLRP and Claim \ref{cl:cc1},
$\frac{f_1(\kappa^*)}{f_2(\kappa^*)} > 1$. This contradicts Claim \ref{cl:cc2}. Hence, an
optimal {\sc post-2} mechanism must be an interior {\sc post-2} mechanism.
Using Claim \ref{cl:interior} proves equations (\ref{eq:k2star}) and (\ref{eq:k1star}).

\subsection{Proof of Theorem \ref{theo:uniform}}
  We prove the three parts of Theorem \ref{theo:uniform}.
  By Theorem \ref{theo:main}, the optimal mechanism is either a {\sc post-1} or {\sc post-2}
  mechanism. By Theorem \ref{theo:post2}, the optimal {\sc post-2} mechanism generates
  weakly higher expected revenue than the optimal {\sc post-1} mechanism (since $b \le \widetilde{\kappa}_1$).

  \noindent {\sc Proof of (1).} Since $b \le \widetilde{\kappa}_1$, the  optimal {\sc post-1} mechanism has a price of $b$ with expected revenue $b(1-F_1(b))$.
  Now, since $b < \widetilde{\kappa}_2$ and MLRP holds, by $\frac{f_1(\widetilde{\kappa}_2)}{f_2(\widetilde{\kappa}_2)} \le 1$,
  we have
  \begin{align}
  f_1(b) \le f_2(b) \label{eq:f12}
  \end{align}
  Now, consider a uniform {\sc post-2} mechanism and the derivative of its expected revenue from (\ref{eq:k11})
  at $\kappa^*=b$:
  \begin{align*}
  -bf_1(b) + (1-F_2(b)) & \ge -bf_2(b) + (1-F_2(b)) > 0
  \end{align*}
  where the first inequality follows from (\ref{eq:f12}) and the second follows from the fact that $F_2$
  satisfies SC and $b < \widetilde{\kappa}_2$.
  So, the expected revenue is increasing in $\kappa^*$ at $b$.
  Hence, the optimal uniform {\sc post-2} mechanism
  has $\kappa^* > b$ and generates more expected revenue than $b(1-F_1(b))$, the optimal {\sc post-1} mechanism.
  By Theorem \ref{theo:th4}, the result then follows.

  \noindent {\sc Proof of (2).} We show that the expected revenues are the same in the
  optimal {\sc post-1} and optimal {\sc post-2} mechanism.
  Assume for contradiction that the optimal {\sc post-2} mechanism generates
  strictly higher expected payoff than the optimal {\sc post-1} mechanism.
  Since $\frac{f_1(\widetilde{\kappa}_2)}{f_2(\widetilde{\kappa}_2)} > 1$ the optimal
  mechanism is an interior {\sc post-2} mechanism by Theorem \ref{theo:th4}. This implies that if the optimal interior {\sc post-2} mechanism is $\kappa_1=\kappa^*_1$ and
  $\kappa_2=\kappa^*_2$, we have $b < \kappa^*_1 < \kappa^*_2 < \beta$.
  By Claim \ref{cl:interior}, $\kappa^*_2=\widetilde{\kappa}_2$. This implies $b < \widetilde{\kappa}_2$, which
  contradicts the assumption of (2).

  Hence, the optimal {\sc post-2} mechanism must generate the same payoff as the
  optimal {\sc post-1} mechanism, which is $b(1-F_1(b))$. This is also the
  revenue of the uniform {\sc post-2} mechanism with price $\kappa^*=b$.
  Thus, the optimal {\sc post-2} mechanism is a uniform {\sc post-2} mechanism
  with price $b$ and generates the same expected revenue as the optimal {\sc post-1} mechanism.

  \noindent {\sc Proof of (3).} We show that the expected revenues are the same in the optimal {\sc post-1} and optimal {\sc post-2} mechanism. Assume for contradiction that the optimal {\sc post-2} mechanism generates
  strictly higher expected payoff than the optimal {\sc post-1} mechanism.
  Since $\frac{f_1(\widetilde{\kappa}_2)}{f_2(\widetilde{\kappa}_2)} \le 1$,
  by Theorem \ref{theo:th4}, the optimal mechanism is a uniform {\sc post-2}
  mechanism. By Claim \ref{cl:cc2}, $\frac{f_1(\kappa^*)}{f_2(\kappa^*)} \le 1$.
  Since $\frac{f_1(b)}{f_2(b)} > 1$ and $\kappa^* > b$, MLRP implies that
  $\frac{f_1(\kappa^*)}{f_2(\kappa^*)} > 1$, a contradiction.
  Hence, the optimal {\sc post-2} mechanism must generate the same payoff as the
  optimal {\sc post-1} mechanism, which is $b(1-F_1(b))$. This is also the
  revenue of the uniform {\sc post-2} mechanism with price $\kappa^*=b$.
  Thus, the optimal {\sc post-2} mechanism is a uniform {\sc post-2} mechanism
  with price $b$ and generates the same expected revenue as the optimal {\sc post-1} mechanism.

\section{Revelation principle}
\label{sec:rev}

In this appendix, we establish a revelation principle which allows us to work
with direct mechanism.
A mechanism is specified by a message space $M$ and an outcome function $\mu:M \rightarrow Z$.
Define the range of the mechanism $(M,\mu)$ as
\begin{align*}
R(M,\mu) &= \{\mu(m): m \in M\}
\end{align*}
Suppose the type space is $V \subseteq [0,\beta]^2$.
A strategy in mechanism $(M,\mu)$ is a map $s: V \rightarrow M$.
Strategy $s$ is an {\bf equilibrium} if for every $v \in V$,
\begin{align*}
\mu(s(v)) \in Ch(R(M,\mu);v)
\end{align*}

The direct mechanism $(V,\mu^*)$ is {\bf incentive compatible} if for each $v \in V$,
\begin{align*}
\mu^*(v) \in Ch(R(V,\mu^*);v)
\end{align*}

For any pair of distinct outcomes $(a,t)$ and $(a',t')$, we say $(a,t)$ transfer-dominates $(a',t')$,
written as $(a,t) \succ_{tr} (a',t')$ if $t > t'$.

\begin{theorem}
  \label{theo:rev}
If $s$ is an equilibrium in $(M,\mu)$, then there exists an incentive compatible direct mechanism $(V,\mu^*)$
such that $\mu^*(v) = \mu(s(v))$ or $\mu^*(v) \succ_{tr} \mu(s(v))$ for all $v \in V$.
\end{theorem}

The proof of this theorem uses a series of claims. 

For any $X \subseteq Z$, let $X^*:=\{(a,t) \in X: (a,t) \in Ch(X;v)~\textrm{for some}~v \in V\}$.
\begin{claim}
\label{cl:cl1}
Suppose $X \subseteq Z$ is such that $Ch(X;v)$ is non-empty for all $v \in V$.
Then,
\begin{align*}
Ch(X;v) &= Ch(X^*;v)~\qquad~\forall~v \in V
\end{align*}
\end{claim}
\begin{proof}
Consider any $v \equiv (v_1,v_2)$ and $X \subseteq Z$ such that $Ch(X;v)$
is non-empty. This implies that $X^*$ is also non-empty. We consider three cases. \\

\noindent {\sc Case 1.} $Ch(X;v_2) \rhd_{v_1} Ch(X_b;v_1)$. Note that in this case,
$Ch(X;v)=Ch(X;v_2)$. Hence, $Ch(X;v_2) \subseteq X^*$.  Since $X^* \subseteq X$, we get $Ch(X^*;v_2)=Ch(X;v_2)$.
Now, $Ch(X_b;v_1) \unrhd_{v_1} Ch(X^*_b;v_1)$ since $X^*_b \subseteq X_b$. Then,
$Ch(X;v_2) \rhd_{v_1} Ch(X_b;v_1)$ implies that $Ch(X^*;v_2) \rhd_{v_1} Ch(X^*_b;v_1)$.
Hence, $Ch(X^*;v) = Ch(X^*;v_2)=Ch(X;v_2)=Ch(X;v)$. \\

\noindent {\sc Case 2.} $Ch(X_b;v_1) \sim_{v_1} Ch(X;v_2)$.  Then $Ch(X;v) = Ch(X_b,v_1)$ which implies $Ch(X_b;v_1) \subseteq X^*$. Hence $Ch(X_b;v_1) \subseteq X^*_b$. This implies that $Ch(X^*_b;v_1)=Ch(X_b;v_1)$ since $X^*_b \subseteq X_b$. Now, $X^* \subseteq X$ implies that $Ch(X;v_2) \unrhd_{v_1} Ch(X^*;v_2)$. Then $Ch(X_b;v_1) \sim_{v_1} Ch(X;v_2)$ implies $Ch(X^*_b;v_1) \unrhd_{v_1} Ch(X^*;v_2)$. Therefore, $Ch(X^*;v_2) \rhd_{v_1} Ch(X^*_b;v_1)$ cannot hold. Hence, $Ch(X^*;v) = Ch(X^*_b;v_1) = Ch(X_b;v_1) = Ch(X;v)$.\\

\noindent {\sc Case 3.} There exists $(a',t') \in Ch(X;v_2)$ such that $v_1 a - t > v_1 a' - t'$ for every $(a,t) \in Ch(X_b;v_1)$. Fix any $(a,t) \in Ch(X_b;v_1)$. If $t' > t$ then $v_1 a - t > v_1 a' - t'$ implies $v_2 a - t > v_2 a' - t'$ since $v_1 \geq v_2$. This contradicts $(a',t') \in Ch(X;v_2)$. Therefore, $t' \leq t \leq b$ which implies $(a',t') \in Ch(X_b;v_2)$.  Now, consider type $v'=(v_2,v_2)$ and observe that $Ch(X;v') = Ch(X_b;v_2)$. Therefore, $Ch(X_b;v_2) \subseteq X^*$ which then implies $(a',t') \in X^*$ since $(a',t') \in Ch(X_b;v_2)$. Hence, $(a',t') \in Ch(X^*;v_2)$ since $X^* \subseteq X$.

Since $Ch(X^*_b;v_1)=Ch(X_b;v_1)$ as in Case 2, we conclude $(a',t') \in Ch(X^*;v_2)$ such that $v_1 a - t > v_1 a' - t'$ for every $(a,t) \in Ch(X^*_b;v_1)$. Then $Ch(X^*;v) = Ch(X^*_b;v_1) = Ch(X_b;v_1) = Ch(X;v)$.
\end{proof}

For every $X \subseteq Z$, an outcome $(a,t)$ is {\bf tr-max} in $X$
if $(a,t) \succ_{tr} (a',t')$ for all $(a',t') \in X \setminus (a,t)$.
Consider any value $v_i \in [0,\beta]$. If $v_i \ne 0$ and a tr-max outcome
exists in $Ch(X;v_i)$, then it is unique. If $v_i=0$, there may be more
than one tr-max outcome in $Ch(X;v_i)$. In that case, we assign one of them
as tr-max arbitrarily. We denote this tr-max outcome at every $v_i \in [0,\beta]$
and every $X$ as $Ch^*(X;v_i)$:
\begin{align*}
Ch^*(X;v_i) &= (a,t)~\qquad~\textrm{if}~(a,t)~\textrm{is tr-max in}~Ch(X;v_i)
\end{align*}
The following claim is useful.
\begin{claim}
\label{cl:nempty}
Suppose $X \subseteq Z$ and $v_i \in [0,\beta]$ be such that $Ch(X;v_i)$
is non-empty. Then, $Ch^*(X;v_i)$ exists.
\end{claim}
\begin{proof}
If $Ch(X;v_i)$ is non-empty, then let $(a,t) \in Ch(X;v_i)$. Let $u_i=a v_i - t$.
Now, consider the following maximization program:
\begin{align*}
\max_{(a',t')} t'~~\textrm{s.t.}~~v_ia' - t' = u_i
\end{align*}
Each outcome in $Ch^*(X,v_i)$ must be a solution to this. But this is equivalent to solving
\begin{align*}
\max_{a' \in [0,1]} [v_ia'-u_i]
\end{align*}
Since this is a maximization of a linear function over a compact set, an optimal solution exists.
This shows that $Ch^*(X;v_i)$ exists.
\end{proof}

We can now define $Ch^*(X;v)$ analogous to $Ch(X;v)$ for every $v \in V$:
\begin{align*}
Ch^*(X;v) =
\begin{cases}
Ch^*(X;v_2) &~\textrm{if}~Ch^*(X;v_2) \rhd_{v_1} Ch^*(X_b;v_1) \\
Ch^*(X_b;v_1) &~\textrm{otherwise}
\end{cases}
\end{align*}
where we abuse notation to write $Ch^*(X;v_2) \rhd_{v_1} Ch^*(X_b;v_1)$
instead of $\{Ch^*(X;v_2)\} \rhd_{v_1} \{Ch^*(X_b;v_1)\}$.  Note that if $Ch^*(X;v_2) \rhd_{v_1} Ch^*(X_b;v_1)$ is not true, then $Ch^*(X_b;v_1) \unrhd_{v_1} Ch^*(X;v_2)$ is true.

We can now state a claim analogous to Claim \ref{cl:cl1}.
For every $X \subseteq Z$, define $\widehat{X}$ as
\begin{align*}
\widehat{X}:= \{Ch^*(X;v):v \in V\}
\end{align*}
\begin{claim}
  \label{cl:cl2}
  Suppose $X \subseteq Z$ is such that $Ch(X;v)$ is non-empty for all $v \in V$.
  Then,
\begin{align*}
  Ch^*(X;v) &= Ch^*(\widehat{X};v)~\qquad~\forall~v \in V
\end{align*}
\end{claim}
\begin{proof}
If $X \subseteq Z$ is such that $Ch(X;v)$ is non-empty for all $v$,
by Claim \ref{cl:nempty}, $Ch^*(X;v)$ exists for all $v$.
Pick $v$ and consider the two possible cases. \\

\noindent {\sc Case 1.} $Ch^*(X;v_2) \rhd_{v_1} Ch^*(X_b;v_1)$. Then, $Ch^*(X;v)=Ch^*(X;v_2)$. Hence, $Ch^*(X;v_2) \in \widehat{X}$.
Since $\widehat{X} \subseteq X$, we get $Ch^*(\widehat{X};v_2)=Ch^*(X;v_2)=Ch^*(X;v)$.
Finally, $Ch^*(X_b;v_1) \unrhd_{v_1} Ch^*(\widehat{X}_b;v_1)$ since $\widehat{X}_b \subseteq X_b$.
So, we have
$Ch^*(X;v_2) \rhd_{v_1} Ch^*(X_b;v_1) \unrhd_{v_1} Ch^*(\widehat{X}_b;v_1)$. Hence,
$Ch^*(\widehat{X};v_2)=Ch^*(X;v_2) \rhd_{v_1} Ch^*(\widehat{X}_b;v_1)$. As a result,
\begin{align*}
Ch^*(\widehat{X};v)=Ch^*(\widehat{X};v_2)=Ch^*(X;v_2)=Ch^*(X;v)
\end{align*}

\noindent {\sc Case 2.} $Ch^*(X_b;v_1) \unrhd_{v_1} Ch^*(X;v_2)$. Then, $Ch^*(X;v)=Ch^*(X_b;v_1)$. Now, let $Ch^*(X_b;v_1) = \{(a,t)\}$ and $Ch^*(X;v_2) = \{(a',t')\}$ and note that $t \leq b$.  Condition of the case implies $v_1a-t \geq v_1a'-t'$. $Ch^*(X;v_2) = \{(a',t')\}$ implies $v_2a'-t' \geq v_2a-t$. Adding the inequalities and using $v_1 \geq v_2$ we derive $a' \leq a$. The second inequality then implies $t' \leq t \leq b$. Therefore, $Ch^*(X;v_2) = Ch^*(X_b;v_2)$. Then consider a type $v' = (v_2,v_2)$ and observe that $Ch^*(X;v') = Ch^*(X;v_2)$. Hence, $Ch^*(X;v_2) \in \widehat{X}$. Since $\widehat{X} \subseteq X$ we then have $Ch^*(X;v_2)=Ch^*(\widehat{X};v_2)$.

Also, $Ch^*(X_b;v_1)=Ch^*(X;v)$ implies $Ch^*(X_b;v_1) \in \widehat{X}$. Hence,
$Ch^*(X_b;v_1) \in \widehat{X}_b$. This implies that $Ch^*(X_b;v_1)=Ch^*(\widehat{X}_b;v_1)$ since
$\widehat{X}_b \subseteq X_b$. As a result, we have $Ch^*(\widehat{X}_b;v_1)=Ch^*(X_b;v_1) \unrhd_{v_1} Ch^*(X;v_2)=Ch^*(\widehat{X};v_2)$.
Hence, $Ch^*(\widehat{X};v)=Ch^*(\widehat{X}_b;v_1)=Ch^*(X_b;v_1)=Ch^*(X;v)$.
\end{proof}

\noindent {\sc Proof of Theorem \ref{theo:rev}}.

\begin{proof}
Let $X:=\{Ch(R(M,\mu);v):v \in V\}$.
Since $s$ is an equilibrium of $(M,\mu)$, $\mu(s(v)) \in Ch(R(M,\mu);v)$.
Hence, $Ch(R(M,\mu);v)$ is non-empty for each $v$.
By Claim \ref{cl:cl1}, $Ch(R(M,\mu);v)=Ch(X;v)$.
So, for each $v$, $Ch(X;v)$ is non-empty, and by Claim \ref{cl:nempty}, $Ch^*(X;v)$ exists.

Now, for every $v \in V$, define
\begin{align*}
\mu^*(v) &:= Ch^*(X;v)
\end{align*}
Since $s$ is an equilibrium in $(M,\mu)$,
for every $v$, we have $\mu(s(v)) \in Ch(R(M,\mu);v) = Ch(X;v)$.
Since $\mu^*(v)= Ch^*(X;v)$, we have $\mu^*(v)=\mu(s(v))$ or
$\mu^*(v) \succ_{tr} \mu(s(v))$.

\noindent Finally, $R(V,\mu^*)=\{\mu^*(v): v \in V\}=\{Ch^*(X;v): v \in V\} = \widehat{X}$.
By Claim \ref{cl:cl2}, for all $v \in V$, we have $Ch^*(X;v)= Ch^*(\widehat{X};v)$, which
further implies that
\begin{align*}
\mu^*(v) = Ch^*(\widehat{X};v) \in Ch(\widehat{X};v),
\end{align*}
which is the required incentive compatibility constraint.
\end{proof}

\end{document}